\documentclass{article}
\usepackage{graphicx} % Required for inserting images

\usepackage[english]{babel}
\usepackage{enumerate}% http://ctan.org/pkg/enumerate
\usepackage{amsmath, amssymb, amsthm,  algpseudocode, algorithm}

\usepackage[a4paper,top=3cm,bottom=3cm,left=3cm,right=3cm,marginparwidth=1.75cm]{geometry}
\usepackage[textsize=small]{todonotes}
\usepackage[round]{natbib}
\usepackage{pifont}
\usepackage{bm}

\usepackage{amsfonts}
\usepackage{bbm}
\usepackage{enumitem}
\usepackage{graphicx}
\usepackage[colorlinks=true, allcolors=blue]{hyperref}
\usepackage{mathtools}
\usepackage{breqn}
\usepackage{float}
\usepackage{xcolor}
\usepackage{tcolorbox}

\newtheorem{theorem}{Theorem}[section]
\newtheorem{lemma}[theorem]{Lemma}

\newtheorem{proposition}[theorem]{Proposition}

\newtheorem{corollary}[theorem]{Corollary}

\newtheorem{remark}[theorem]{Remark}

\newtheorem{assumption}[theorem]{Assumption}

\theoremstyle{definition}
\newtheorem{example}[theorem]{Example}
\newtheorem{definition}[theorem]{Definition}

\numberwithin{equation}{section}

\newcommand{\upperRomannumeral}[1]{\uppercase\expandafter{\romannumeral#1}}
\newcommand{\lowerRomannumeral}[1]{\lowercase\expandafter{\romannumeral#1}}

\newcommand{\R}{{\mathbb R}}

\usepackage{mathtools}
\usepackage{authblk}
\mathtoolsset{showonlyrefs}

\title{Simulating integrated Volterra square-root processes and Volterra Heston models via Inverse Gaussian}
\author{Eduardo Abi Jaber\thanks{eduardo.abi-jaber@polytechnique.edu. I am grateful for the financial support from the Chaires FiME-FDD, Financial Risks, Deep Finance \& Statistics and Machine Learning and systematic methods in finance at Ecole Polytechnique.\\
We would like to thank  Shaun Li and Dimitri Sotnikov for their valuable feedback.} }

\author{Elie Attal} 

\affil{Ecole Polytechnique, CMAP}

\begin{document}

\maketitle

\begin{abstract}
We introduce a novel simulation scheme, iVi (integrated Volterra implicit), for integrated Volterra square-root processes and Volterra Heston models based on the Inverse Gaussian distribution. The scheme is designed to handle $L^1$ kernels with singularities by relying solely on integrated kernel quantities, and it preserves the non-decreasing property of the integrated process. We establish weak convergence of the iVi scheme by reformulating it as a stochastic Volterra equation with a measure kernel and proving a stability result for this class of equations. Numerical results demonstrate that convergence is achieved with very few time steps. Remarkably, for the rough fractional kernel, unlike existing schemes, convergence seems to improve as the Hurst index $H$ decreases and approaches $-1/2$.
\end{abstract}

\section*{Introduction}

The Volterra square-root process $V$, defined by  
\begin{align}\label{eq:V}  
V_t = g_0(t) + b \int_0^t K(t-s) V_s \,ds + c\int_0^t K(t-s) \sqrt{V_s} \,dW_s,  
\end{align}  
where $g_0$ is a suitable deterministic curve, {$b \leq 0$\,, $c \geq 0\,$}, and $K$ is a {nonnegative} locally square-integrable kernel, is a ubiquitous process that appears in diverse contexts. In particular, it arises as a scaling limit of branching processes in population genetics \citep*{dawson1994super, mytnik2015uniqueness} and of self-exciting Hawkes processes in mathematical finance \citep*{jaisson2016rough, el2019characteristic}. 

However, simulating $V$ is intricate due to several factors: the process is in general non-Markovian, not necessarily a semimartingale and involves a square-root nonlinearity. An additional challenge arises when $K$ is singular at the origin, as is the case for the fractional kernel 
\begin{align}\label{eq:fractionalkernel}
 K_H(t) = t^{H-1/2}, \quad 
\end{align}
with $H \in (0,1/2)$. This singularity requires special care in numerical implementations to ensure accuracy and stability.  

We design a simple and efficient simulation scheme for the integrated Volterra  square-root process, leveraging the affine structure, which overcomes all of these difficulties while maintaining numerical robustness with very little time steps.     

A key financial motivation for studying the Volterra square-root process is its role in modeling the variance process of a stock price in the class of so-called Volterra Heston models \cite*[Section 7]{abijaber2019affine}, where the stock price $S$ follows  
\begin{align*}  
dS_t = S_t \sqrt{V_t} \, dB_t\,,  
\end{align*}  
with $B$ being a Brownian motion correlated with $W$.   The kernel $K$ introduces flexibility and encompasses various popular models: the case $K \equiv 1$ yields the celebrated  \cite*{heston1993closed} model, the fractional kernel $K_H$ as in \eqref{eq:fractionalkernel} with a Hurst index $H \in (0,1/2)$ leads to the rough Heston model of \cite*{el2019characteristic}, and  
   the weighted sum of exponentials $K(t) = \sum_{i=1}^N c_i e^{-x_i t}$, $c_i,x_i\geq 0$, 
    $N \in \mathbb N$, gives the  lifted Heston model of \cite*{abi2019lifting}.  An important feature of this class of models is its affine structure which ensures that the characteristic function of the log price is known up to a deterministic    Riccati Volterra equation, enabling pricing and hedging via Fourier inversion techniques.

Our core idea is to simulate the non-decreasing integrated process   
$U_{0,t} = \int_0^t V_s ds$, 
instead of the process $V$ itself, in the spirit of the {integrated Variance implicit} (iVi) scheme of \cite*{abijaber2024simulation}  recently introduced  for the case  $b=0$ and $K(t) = e^{\lambda t}$. Integrating the dynamics \eqref{eq:V} and applying the stochastic Fubini theorem yields the dynamics  
\begin{align}\label{eq:U0}  
U_{0,t} = \int_0^t g_0(s) \, ds + \int_0^t K(t-s)\left(b U_{0,s}  +  c  Z_{0,s} \right) ds,  
\end{align}  
where $Z_{0,\cdot} := \int_0^{\cdot}\, \sqrt{V_s}\,dW_s$ is a continuous martingale with quadratic variation $U$, so that by the Dambis, Dubins-Schwarz theorem we can write it as a time changed Brownian motion $Z = \widetilde W_U$.   

One major advantage of this formulation is that the equation remains valid for kernels $K$ that are only locally integrable, but not necessarily locally square-integrable, see \cite{abi2021weak}. For instance, it holds for the fractional kernel  $K_H$ in \eqref{eq:fractionalkernel},  
even for negative Hurst indices $H \in (-1/2, 1/2]$, see \cite*{abi2021weak, jusselin2020no}.   When $H \in (-1/2,0]$, the kernel $K_H$ is no longer in $L^2$ but remains in $L^1$, meaning that $U$ is still well-defined. However, in this regime, $U$ is no longer absolutely continuous with respect to the Lebesgue measure, and the instantaneous variance process $V$ itself no longer makes sense.  

In the sequel, we develop a simple numerical scheme to simulate $U$ in \eqref{eq:U0}  for any kernel $K \in L^1([0,T])$, without attempting to reconstruct $V$. A key challenge is to construct a scheme that preserves the non-decreasing nature of $U$.  

\textbf{Main contributions.}  Our approach builds upon the increments of $U$ and the corresponding time-changed Brownian motion $\widetilde W_U$, defined as  
\begin{align}\label{eq:UZ}  
U_{s,t} :=  U_{0,t} - U_{0,s}, \quad Z_{s,t} := \widetilde W_{U_{0,t}}-\widetilde W_{U_{0,s}} = \int_s^t d \widetilde W_{U_{0,r}}, \quad s \leq t. 
\end{align}  
These quantities capture the essential information required for many financial applications and are critical in the simulation of other processes, such as the Volterra Heston model.

We already have all the ingredients needed to present the discretization scheme we propose, which emphasizes the simplicity of the approach.  The following algorithm recursively constructs the increments $ (\widehat{U}_{i, i+1})_{i=0, \ldots, n-1} $, and $ (\widehat{Z}_{i, i+1})_{i=0, \ldots, n-1} $. We also coin it the \textbf{iVi} scheme for \textbf{i}ntegrated \textbf{$\bold V$} (\textbf{V}ariance or \textbf{V}olterra) \textbf{i}mplicit scheme as it collapses with the original iVi scheme of \cite{abijaber2024simulation} when $b=0$ and $K(t) = e^{\lambda t}$.

\begin{algorithm}[H]
\caption{ \textbf{- The iVi scheme:} Simulation of \( \widehat U, \widehat Z \)}\label{alg:simulation}
\begin{algorithmic}[1]
\State \textbf{Input:}  parameters {$b\,, c\,,$} \( g_0, K \), uniform partition of $[0,T]$, \( t_i = i T/n  \), with $i=0,\ldots, n$.
\State \textbf{Output:} \(  \widehat U_{i, i+1}, \widehat Z_{i, i+1} \) for \( i = 0, \ldots, n-1 \).
\State Compute
\begin{align}\label{eq:kij}
k_{\ell} = \int_0^{\frac T n} K\left(  \frac{\ell T}{n} + s\right) \, ds, \quad {\ell=0,\ldots, n-1}.
\end{align}
\For{$i = 0$ to $n-1$}

\State Compute the quantity:
\begin{align}\label{eq:alphai}
    \alpha_{i}  = \int_{t_i}^{t_{i+1}} g_0(s) ds + \sum_{j=0}^{i-1}   k_{i-j}\left(b  \widehat U_{j,j+1} +  c \widehat Z_{j,j+1}\right),  
\end{align}
with  $\alpha_0=\int_0^{t_1} g_0(s) ds,$  for $i=0$.
    \State Simulate the increment:
\begin{align}\label{eq:hatUsample}
   \widehat U_{i, i+1} \sim IG\left(\frac{\alpha_i}{1- b k_0}, \left(\frac{\alpha_i}{ck_{0}}\right)^2 \right).
     \end{align}
    \State Set:
    \begin{align}\label{eq:Zii}
    \widehat Z_{i, i+1} = \frac{1}{ck_{0}} \left((1-bk_{0})\widehat U_{i, i+1} - \alpha_i\right).  
    \end{align}
\EndFor
\end{algorithmic}
\end{algorithm}
 Here, $IG$ refers to an Inverse Gaussian distribution, whose definition and a simple sampling algorithm are provided in Appendix~\ref{A:IG}.

Aside from the inherent simplicity of Algorithm~\ref{alg:simulation}, which is both efficient and free from precomputations and fine-tuning of hyperparameters, we show that the  iVi scheme enjoys the following key features:
\begin{itemize}
    \item The scheme is built from straightforward right-endpoint Euler-type discretization rule applied to the dynamics of $U$ and leverages the affine structure on the level of the dynamics of $(U,Z)$. 
    \item The scheme is well-defined, i.e.~$\alpha_i \geq 0$,  and  ensures the  non-decreasing property of $\widehat U$, i.e., $\widehat{U}_{i,i+1} \geq 0$ for all $i = 0, \ldots, n$, as established in Theorem~\ref{T:nonnegative}. To  prove this, we build on an elegant idea introduced by \cite{alfonsi2025nonnegativity} that  reformulates the problem in continuous time and leverages kernels that preserve nonnegativity. Intuitively, such kernels guarantee that if a discrete convolution remains nonnegative at specific discretization points, it stays nonnegative for all times. Unlike \cite{alfonsi2025nonnegativity}, where this property is imposed on the kernel $K$, we require it for  integrated kernels as in \eqref{eq:barGbarK}, which is also more accurate numerically, especially for singular kernels. Moreover, the simulated process $\widehat Z$ remains a martingale for finite $n\,$.
    \item The piecewise constant processes $\left( \widehat U\,, \widehat Z\right)$ generated by the scheme are tight in the Skorokhod $J_1\left(\mathbb{R}^2\right)$ topology on $[0\,,T]\,$, and any accumulation point satisfies the stochastic Volterra equation \eqref{eq:U0}. Under weak uniqueness of solutions, this guarantees the convergence of the scheme as the number of discretization steps $n$ tends to $\infty$, see Theorem \ref{theorem:weak_convergence} and Corollary \ref{corollary:weak_convergence_monotone}. The proof relies on properties of the Inverse Gaussian distribution along with a reformulation of the scheme in terms of a stochastic Volterra equation with a measure-valued kernel, see \eqref{eq:convolution_equation_measure_kernel}.~This approach is crucial, as it allows us to rigorously handle the non-Markovian  structure of the problem and obtain convergence by studying the stability of affine Volterra equations. In particular, our approach leads to a new weak existence result for solutions of \eqref{eq:U0}    under slightly different assumptions than \citet[Theorem 2.13]{abi2021weak}.
    \item The scheme captures essential distributional properties, which are clearly observed numerically when comparing its performance for options pricing with other schemes (see Figures~\ref{fig:U} and  \ref{fig:call}):
    \begin{itemize}
        \item[(i)] The Inverse Gaussian distribution emerges naturally from the right-endpoint discretization of the conditional characteristic function of $U$, as shown in Remark~\ref{R:charcomparison}.
        \item[(ii)] The scheme accurately reproduces the Inverse Gaussian limiting distribution of $U_{0,T}$ as $H\to -1/2$, as recently established in \cite*{jaber2025hyper}.
         \end{itemize}
    \item In terms of performance, the iVi scheme displays high precision with very few time-steps, for the integrated Volterra square-root process and the Volterra Heston model,  even under challenging regimes with low and negative Hurst indices when the fractional kernel $K_H$ in \eqref{eq:fractionalkernel} is considered, as illustrated on Figure~\ref{fig:surface}.~More remarkably, we observe that as $H$ decreases to  $ -1/2$ our iVi scheme seems to converge faster and with less time steps. 
\end{itemize}

\begin{figure}[h!]
    \centering
\includegraphics[width=1.\textwidth]{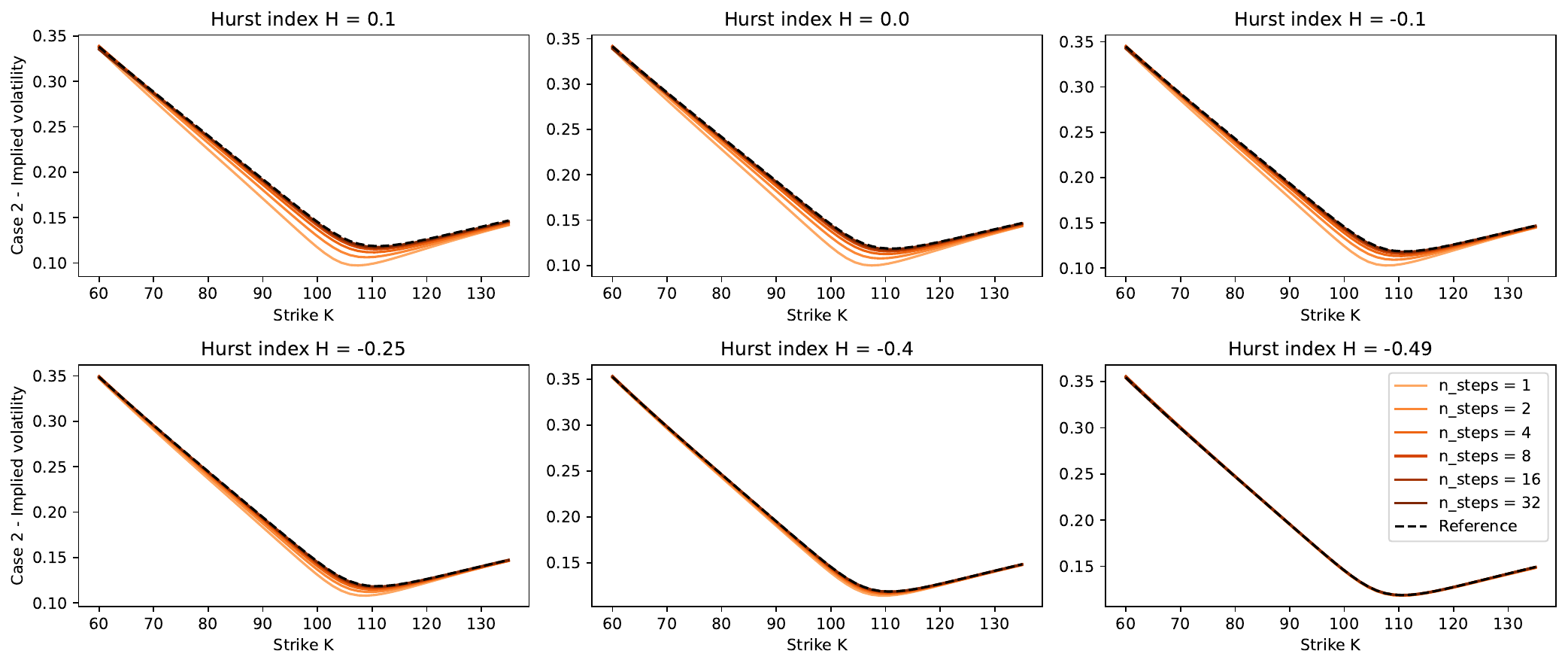} \\
    \caption{Implied volatility slice for call options of the hyper-rough Heston model \eqref{eq:HestonS}, with maturity $T=1$. Parameters as in Case 2 of Table~\ref{tab:parameter_cases}, with varying Hurst index between $0.1$ and -$0.49$. iVi scheme in shades of orange with the number of corresponding time steps with 1 million sample paths.}
    \label{fig:surface} % Optional: for referencing the figure
\end{figure}

 \textbf{Relation to the Literature.}  Several works in the literature address the simulation of Volterra square-root processes and Volterra Heston models, with a particular focus on the rough case. 
 
 A common approach relies on approximating a completely monotone kernel, such as the fractional kernel $K_H$ in \eqref{eq:fractionalkernel}, by a weighted sum of exponentials. One of the earliest contributions is the multi-factor Markovian approximation of \cite{abi2019multifactor},  which combines this approximation with a  clipped Euler-type scheme, see \cite[Appendix A.4]{abi2019lifting}. While it ensures positivity through  clipping and might be applied with $L^1$ kernels, no theoretical convergence result is provided. Subsequent works such as those of  \cite{alfonsi2024approximation,alfonsi2025nonnegativity,bayer2024efficient}  refined the schemes based on the same multi-factor approximation. In particular, \cite{bayer2024efficient} design a low-dimensional Markovian approximations with an adaptation of the moment-matching scheme of the conventional square-root process of \cite{lileika2020weak}.  \cite*{abijaber2024state} recently established the well-posedness of the scheme by characterizing the state space of the Markovian factors. While this method is well-optimized for the fractional kernel, works for integrable kernels,  it lacks a theoretical convergence guarantee and deteriorates as $H$ decreases.  Additionally, \cite{alfonsi2025nonnegativity} proposed a splitting scheme with theoretical convergence guarantees and positivity preservation, but it is restricted to square-integrable kernels.~Regardless of their precision, all these methods introduce an approximation error due to the replacement of the original kernel with a sum of exponentials.

An alternative approach directly tackles the stochastic Volterra equation. \cite*{richard2023discrete} studied explicit Euler-type discrete-time schemes for the rough Heston model, applied both to the Volterra process $V$ and the integrated process $U$. Their integrated scheme  works for kernels in $L^1$ but preserves positivity only via  clipping. Interestingly, their Equation (9) introduces an explicit scheme at the level of $U$, but  clipping is necessary to ensure monotonicity. For comparison with our implicit iVi scheme, we implement  such explicit scheme in \eqref{eq:explicitscheme}.

Finally, \cite{gatheral2022efficient} adapted the widely used Quadratic Exponential scheme of \cite{andersen2007efficient} to the rough Heston model and more general affine rough forward rate models. However, this method requires square integrability of the kernel, lacks a convergence result, and its accuracy deteriorates as $H$ decreases.

Compared to the existing literature, our iVi scheme is the first to combine the following key features. It effectively handles $L^1$ kernels with singularities by relying solely on integrated quantities of the kernel \eqref{eq:kij}, ensures the non-decreasing property of the integrated process $\widehat{U}$ without requiring additional clipping, provides a theoretical guarantee of convergence, and remains remarkably simple. Its implementation is nearly effortless, requiring neither moment matching nor costly pre-computations. Furthermore, numerical results indicate that convergence is achieved with very few time steps, and, unlike all other schemes, it improves as $H$ approaches $-1/2$, a phenomenon that can be theoretically justified by the limiting Inverse Gaussian regime as $H \to -1/2$.

A key advantage of having a stable scheme with minimal time steps is that, due to the non-Markovian nature of the process, computational complexity grows quadratically with the number of discretization points. Reducing the number of time steps is therefore essential for ensuring efficiency in practical implementations.

{Finally, we mention that the Inverse Gaussian distribution is intimately linked to Heston-type models, having emerged as a limiting law in various asymptotic regimes. In the classical Heston model, \cite{forde2011large} established that the long-time distribution of the integrated variance converges to an Inverse Gaussian distribution.~Building on this, \cite{mechkov2015fast} introduced fast mean-reversion and volatility-of-volatility scalings to attain such a regime, a framework further developed by \cite{McCrickerd2021}. More recently,  using such scalings, \cite{abijaber2024reconciling} investigated an indirect connection between the rough Heston model and the Inverse Gaussian distribution as the Hurst parameter $H\to -1/2$.
 This limiting behavior has since been rigorously established by \cite*{jaber2025hyper}.}

\textbf{Outline.} In Section~\ref{S:scheme} we provide the mathematical derivation of the scheme as well as its distributional properties.~Weak convergence of the scheme is presented in Section~\ref{S:convergence}.~Numerical illustrations for the scheme for the integrated process $U$ and the Heston model are provided in Section \ref{section:numerical}. The proof of the well-definedness of the scheme is postponed to Section~\ref{S:ProofMainWellDefined} and that of the convergence to Section~\ref{S:ProofConvergence}.~Useful properties of the Inverse Gaussian distribution and the  Heston model are collected in the appendices.

\section{The iVi scheme}\label{S:scheme}
 In this section, we detail the mathematical derivation of the iVi scheme, prove that it  preserves nonnegativity for the discretized process $\widehat V$ and relate it to distributional properties through the characteristic function.

We consider a non-decreasing Volterra process $U$ of the form
\begin{align}
U_{0,t} = \int_0^t g_0(s) \, ds + \int_0^t K(t-s)\left(b U_{0,s}  +  c Z_{0,s} \right) ds  
\end{align}
where  $K$ is a nonnegative locally integrable kernel ($L^1([0,T])$, not necessarily in $L^2$), $b\leq 0$ and $c \geq 0\,$,  $Z$ is a continuous martingale with quadratic variation $U$ and $g_0$ is a {nonnegative locally integrable input curve}.

\subsection{Deriving the iVi scheme}
The first step is to write the dynamics of the increments  $U_{s,t} = U_{0,t} - U_{0,s}$  between $s$ and $t$, in terms of a $\mathcal F_s$-measurable part summarizing the non-Markovianity of the system and an increment between $s$ and $t$. For this, we define, for  $s\leq t$, 
\begin{align}\label{eq:Gs}
    G_s(t) &=  \int_s^t g_0(u) \, du + \int_0^s \int_s^t K(u-r) \,du \, (b \, dU_{0,r} + c \, dZ_{0,r})
\end{align}
and note that $G_s(t)$ is $\mathcal F_s$-measurable for each $t \geq s$.  We also denote by  $dU_{s,u}$, $dZ_{s,u}$  the differentials with respect to the second variable $u$, for $u\geq s\,$. 

\begin{proposition}  The dynamics of the increments of $U$ defined in \eqref{eq:UZ} are given by 
\begin{align}\label{eq:Ust}
    U_{s,t}=    G_s(t) 
+ \int_s^t K(t-r) \left( b U_{s,r} + c Z_{{s,r}}\right)  \,dr, \quad s\leq t,
\end{align}
where $G$ is defined by \eqref{eq:Gs}. 
\end{proposition}

\begin{proof} We define $\widetilde Z$ by 
\begin{align}\label{eq:tildeZ}
    \widetilde Z_{s,r} = b U_{s,r} + c Z_{s,r}, \quad s \leq r.
\end{align}
Fix $s\leq t$. The dynamics of $U_{0,\cdot}$ in \eqref{eq:U0} yield 
\[
   \begin{aligned}
   U_{s,t} &= U_{0,t} - U_{0,s} \\
   &= \int_s^t g_0(u) \,du  + \int_0^s (K(t-r)-K(s-r)) \widetilde Z_{0,r} dr + \int_s^t K(t-r) \widetilde Z_{0,r} \, dr\\
   &= \int_s^t g_0(u) \,du  + \int_0^s (K(t-r)-K(s-r)) \widetilde Z_{0,r} dr + \int_s^t K(t-r)  \, dr \widetilde Z_{0,s} +  \int_s^t K(t-r) \widetilde Z_{s,r} \, dr\,,
   \end{aligned}
\]
where we used that $\widetilde Z_{0,r} = \widetilde Z_{0,s} + \widetilde Z_{s,r}$, for $r\geq s$, for the last equality. Hence,  in order to obtain  \eqref{eq:Ust}, it suffices to prove that 
   \begin{align}\label{eq:tempfubini} 
     \int_0^s \int_s^t K(u-r) \,du \, d\widetilde Z_{0,r} =   \int_0^s (K(t-r)-K(s-r)) \widetilde Z_{0,r} dr + \int_s^t K(t-r)  \, dr \widetilde Z_{0,s}. 
   \end{align}
   This is done by 
successive applications of stochastic Fubini's theorem and a change of variables as follows
\begin{align*}
     \int_0^s \int_s^t K(u-r) \,du \, d\widetilde Z_{0,r} &= \int_0^s \int_{s-r}^{t-r} K(v) \,dv \, d\widetilde Z_{0,r} \\
   &=  \int_0^t  K(v) \int_{0}^{t} 1_{\{ 0\leq r \leq s \}} 1_{\{ s-v \leq r\leq t-v\}} \, d\widetilde Z_{0,r}  \,dv  \\
   &= \int_0^t K(v) \left(\widetilde Z_{0, s\wedge (t-v)} - \widetilde Z_{0, 0 \vee (s-v)}\right) \,dv \\
   &= \int_0^{t-s} K(v) \,dv \widetilde Z_{0,s} + \int_{t-s}^t K(v) \widetilde Z_{0,t-v} \,dv \\
   &\quad - \int_0^s K(v) \widetilde Z_{0,s-v}\,dv - \int_s^t K(v) \,dv \widetilde Z_{0,0}\\
   &= \int_s^t K(t-r) \,dr \widetilde Z_{0,s}  + \int_0^s K(t-r) \widetilde Z_{0,r}\, dr - \int_0^s K(s-r) \widetilde Z_{0,r}dr,  
\end{align*}
where we used that $\widetilde Z_{0,0}=0$. This yields \eqref{eq:tempfubini} and ends the proof. 
\end{proof}

The main idea behind our scheme is to simply discretize the equation \eqref{eq:Ust} for the increments of the process $U$ between $t_{i}$ and $t_{i+1}$ in an implicit way. 
First we discretize $G_{t_{i}}(t_{i+1})$. The second term in $G_t(s)$  is discretized as follows, recalling the definition of $\widetilde Z$ in \eqref{eq:tildeZ} and  using that $d\widetilde Z_{0,r}=d\widetilde Z_{t_j,r}$ for $r\geq t_j$,  
\begin{align}
    \int_0^{t_i} \int_{t_i}^{t_{i+1}} K(u-r) \,du \, d\widetilde Z_{0,r} &=  \sum_{j=0}^{i-1}  \int_{t_j}^{t_{j+1}} \int_{t_i}^{t_{i+1}} K(u-r) \,du \, d\widetilde Z_{t_j,r} \\
    &\approx  \sum_{j=0}^{i-1}    \int_{t_i}^{t_{i+1}} K(u-t_j) \,du \int_{t_j}^{t_{j+1}} \, d\widetilde Z_{t_j,r} \\
    &= \sum_{j=0}^{i-1}    \int_{t_i}^{t_{i+1}} K(u-t_j) \,du  \widetilde Z_{t_j,{t_{j+1}}},
\end{align}
which leads, with a uniform partition $t_i=iT/n$ for $i=0,\ldots, n$,  to
\begin{align}\label{eq:Gapprox}
    G_{t_i}(t_{i+1}) \approx  \int_{t_{i}}^{t_{i+1}} g_0(u) \, du + \sum_{j=0}^{i-1}    \int_{t_i}^{t_{i+1}} K(u-t_j) \,du  \,\widetilde Z_{t_j,{t_{j+1}}} = \alpha_i,
\end{align}
with $\alpha_i$ given by \eqref{eq:alphai}. 
As for the second term on the right hand side of \eqref{eq:Ust},  using the right endpoint rule on  $\widetilde Z_{t_i,\cdot}$, i.e.~approximating $\widetilde Z_{t_i,s} $ by the value $\widetilde Z_{t_i,t_{i+1}}$ for $s\in [t_i,t_{i+1})$, yields
\begin{align}
\int_{t_i}^{t_{i+1}} K(t_{i+1}-r)  \widetilde Z_{t_i,r} \, dr \approx \int_{t_i}^{t_{i+1}} K(t_{i+1}-r)    \, dr \widetilde Z_{{t_i,t_{i+1}}} = k_0 \widetilde Z_{{t_i,t_{i+1}}},
\end{align}
with $k_0$ as in \eqref{eq:kij}.
Plugging the two approximations above in \eqref{eq:Ust} and recalling \eqref{eq:tildeZ}, yields 
\begin{align}\label{eq:Uti}
     U_{{t_i},{t_{i+1}}}  \approx  \alpha_i + k_0 bU_{{t_i,t_{i+1}}} +  k_0 c Z_{{t_i,t_{i+1}}},
\end{align}
where the $\left(\alpha_i\right)_{i \leq n-1}$ are defined in \eqref{eq:alphai}.
Now the key point is to observe that having contructed $ U_{t_j,t_{j+1}},Z_{t_{j},t_{j+1}}$ for $j=0,\ldots, i-1$, this forms an implicit equation on $ U_{{t_i},{t_{i+1}}}$ by rewriting $Z$ using a time-changed Brownian motion. Indeed, recall from the definition of $(Z_{t_i,s})_{s\geq t_i}$ in \eqref{eq:UZ} that it is a  continuous local martingale with quadratic variation   $U_{t_{i},\cdot}$. Hence, an application of the celebrated Dambis, Dubins-Schwarz Theorem, see \cite[Theorem 1.1.6]{revuz2013continuous}, yields the representation of $(Z_{t_i,s})_{s\geq t_i}$ in terms of a time changed Brownian motion $Z_{t_i, s} = \widetilde W_{U_{t_i,s}}$ for all $s\geq t_{i}$ where $\widetilde W$ is a standard Brownian motion. 

Hence, plugging this in \eqref{eq:Uti}, approximating $U_{t_i, t_{i+1}}$  boils down to finding  a nonnegative random variable $\widehat U_{i,i+1}$ solving  
\begin{align}\label{eq:IGequation}
   (1-b k_0) \widehat U_{i,i+1} = \alpha_i + k_{0} c\widetilde W_{\widehat U_{i,i+1}}.
\end{align}  

Said differently,  $\widehat U_{i,i+1}$ is a passage  time of the level $\alpha_i$ for the drifted Brownian motion $ ((1-b k_0)s -  k_{0}c \widetilde W_{s})_{s\geq 0} $. In particular,  the first passage time 
$$ X = \inf \left\{ s\geq 0: (1-b k_0) s -  c k_{0}\widetilde W_{s}= \alpha_i \right\} $$
satisfies \eqref{eq:IGequation}. 
It is well known that $X$ follows an  Inverse Gaussian distribution $IG(\mu_i,\lambda_i)$ with mean parameter $\mu_i= \frac{\alpha_i}{1-b k_0}$  and shape parameter $\lambda_i = \frac{\alpha_i^2}{k_{0}^2 c^2}$. Inverse Gaussian distributions are recalled in Appendix~\ref{A:IG}.  

Hence, we sample $\widehat U_{i,i+1}$ using  the Inverse Gaussian distribution as in \eqref{eq:hatUsample}. Beyond its tractability and efficient  sampling, see Algorithm~\ref{alg:IG_sampling}, the choice of the Inverse Gaussian distribution is further justified in Section~\ref{S:whyIG} where it is shown to be intimately linked with the distributional properties of the integrated square-root process $U$. Then, using \eqref{eq:IGequation} we can define the   $\widehat Z_{i,i+1}$ (which plays the role of $\widetilde W_{\widehat U_{i,i+1}}$) using equation \eqref{eq:Zii}. 

Putting everything together, we arrive at the iVi scheme of Algorithm~\ref{alg:simulation}.

\subsection{Well-definedness of the scheme}

To ensure the well-definedness of Algorithm~\ref{alg:simulation}, we need to verify that the mean parameter $\alpha_i$ in \eqref{eq:alphai} of the Inverse Gaussian distribution is nonnegative for each $i=0,\ldots,n$, with the convention that \( IG(0,0) = \delta_0 \) the dirac at $0$.

This is the objective of Theorem~\ref{T:nonnegative}. As a by-product, we will show that the discretized increments $\widehat{U}_{i,i+1}$ are nonnegative for all \( i = 0, \ldots, n-1 \).

The proof relies on a continuous time reformulation of the problem.~For this, we introduce the functions $\bar{g}_0$ and $\bar{K}$ defined by:
\begin{align}\label{eq:barGbarK}
    \bar{g}_0(t) = \int_0^{\frac{T}{n}} g_0(t+s) \, ds \quad \text{and} \quad \bar{K}(t) = \int_0^{\frac{T}{n}} K(t+s) \, ds,
\end{align}
and observe that $\bar{K}(t_i - t_j) = k_{i-j}$, so that $\alpha_i$ in \eqref{eq:alphai} can be rewritten as:
\[
\begin{aligned}\label{eq:alphai2}
    \alpha_i = \bar{g}_0(t_i) + \sum_{j=0}^{i-1} \frac{\bar{K}(t_i - t_j)}{\bar{K}(0)} \left( \widehat U_{j,j+1} - \alpha_j \right).
\end{aligned}
\]
One approach is to prove that $\alpha_i$ remains nonnegative recursively, but this would result in cumbersome recursive formulas involving the kernel \( K \) and the input function \( g_0 \). A more elegant solution simplifies these computations by leveraging an  idea introduced in \cite{alfonsi2025nonnegativity}. This approach reformulates the problem ``continuously in time"  and utilizes kernels that preserve nonnegativity, as described in \cite[Definition~2.1]{alfonsi2025nonnegativity}.

Intuitively, kernels that preserve nonnegativity ensure that if a  discrete convolution  by the kernel is nonnegative when evaluated at some  discretization points, then, the discrete convolution  remains nonnegative over all time.

\begin{definition}
    Let \( K \colon \mathbb{R}_+ \to \mathbb{R}_+ \) be a kernel such that \( K(0) > 0 \). We say that the kernel \( K \) preserves nonnegativity if, for any \( i \geq 1 \), any \( x_1, \ldots, x_i \in \mathbb{R} \), and \( 0 \leq t_1 < \cdots < t_i \) such that:
    $$
    \sum_{j=1}^\ell K(t_\ell - t_j) x_j \geq 0, \quad \ell = 1, \ldots, i,
    $$
    we have:
    $$
    \sum_{j : t_j \leq t} K(t - t_j) x_j \geq 0, \quad t \geq 0,
    $$
    with the convention \( \sum_{\emptyset} = 0 \).
\end{definition}

Unlike \cite{alfonsi2025nonnegativity}, where the nonnegative preserving property is required for the kernel \( K \), we impose this requirement on  the integrated kernel \( \bar{K} \), to be able to adapt the ideas to our context.

\begin{theorem}\label{T:nonnegative}
    Let $K, g_0 \in L^1_{\rm loc}(\R_+, \R)\,$, $b \leq 0$ and $c \geq 0\,$.  Assume that in \eqref{eq:barGbarK},  $\bar g_0:\R_+\to \R_+$ is non-decreasing and nonnegative  and that $\bar K:\R_+ \to \R_+$ is non-increasing and preserves nonnegativity.~Consider   \( ( \widehat U_{i, i+1})_{i=0, \ldots, n-1} \), and \( ( \widehat Z_{i, i+1})_{i=0, \ldots, n-1} \) satisfying the recursions of the {\normalfont iVi} scheme in  Algorithm~\ref{alg:simulation}. Then,  we have that 
    \begin{align}
         \alpha_i, \widehat U_{i,i+1} \geq 0, \quad i=0,\ldots, n-1,
    \end{align}
and the algorithm is well-defined.
\end{theorem}

\begin{proof}
    See Section~\ref{S:ProofMainWellDefined}.
\end{proof}

The following proposition provides examples of input curves $g_0$ and kernels $K$ that satisfy the assumptions of Theorem~\ref{T:nonnegative}.    

\begin{proposition}
Let $K, g_0 \in L^1_{\rm loc}(\mathbb{R}_+, \mathbb{R})$ such that $g_0$ is non-decreasing and nonnegative, and $K$ is non-increasing with the following representation 
\begin{align}\label{eq:laplacetranform}
    K(t) = \int_{\R} e^{- x t} \mu(dx), \quad t >0, 
\end{align}
where $\mu$ is  a 
measure of locally bounded variation.~Then, the function $\bar g_0: \mathbb{R}_+ \to \mathbb{R}_+$ is non-decreasing and nonnegative, while $\bar K: \mathbb{R}_+ \to \mathbb{R}_+$ is non-increasing and preserves nonnegativity.  
\end{proposition}

\begin{proof}
    The part for $\bar g_0$ is straightforward.~Note that since $K$  is non-increasing so is  $\bar K$.~It remains to  argue that $\bar K$ preserves nonnegativity. For this, it suffices to observe that, thanks to Fubini's theorem, $\bar K$ admits the representation 
    $ \bar K(t) = \int_{\R_+}e^{-xt} \nu(dx) $, for $t \geq 0$, with $\nu(dx) = ((1 - e^{xT/n})/x)\mu(dx)$. An application of \cite[Theorem 2.11]{alfonsi2025nonnegativity} yields that $\bar K$ preserves nonnegativity. 
\end{proof}

\begin{example}
\label{example:completely_monotone}
    Any completely monotone kernel $K$, i.e. a function that is infinitely differentiable on $(0, \infty)$ and satisfies  
    $(-1)^j K^{(j)}(t) \geq 0$ for all $t > 0$ and $j \geq 0$, satisfies \eqref{eq:laplacetranform} by Bernstein's theorem, with $\mu$ a nonnegative measure on $\mathbb{R}_+$.  

    Examples of such kernels include:  
    \begin{itemize}
        \item Weighted sum of exponentials: 
        $K(t) = \sum_{i=1}^N c_i e^{-x_i t}, \quad c_i \geq 0, \quad x_i \geq 0, \quad N \in \mathbb{N}.$
        \item Fractional kernel:  
        $K(t) = c t^{H - 1/2}, \quad c > 0, \quad H \in (-1/2, 1/2].$
        \item Shifted fractional kernel: 
        $K(t) = c (t+\epsilon)^{H - 1/2}, \quad c > 0, \quad \epsilon > 0, \quad H \leq 1/2.$
    \end{itemize}
    In addition, sums and products of completely monotone kernels are again completely monotone. 
    Hence, these kernels satisfy the assumptions of Theorem~\ref{T:nonnegative} and with $g_0$ non-decreasing, for instance with $g_0$ given by 
    $$g_0(t) = V_0 + a\int_0^t K(s)\,  ds, \quad t \geq 0,   $$ with $a, V_0 \geq 0$.\qed 
\end{example}

Figure~\ref{fig:samplepaths} presents sample paths of $\left(\widehat{U}, \widehat{Z}\right)$ generated using the iVi scheme in Algorithm~\ref{alg:simulation}, along with the discrete derivative  
$$\widehat{V}_i = \frac{\widehat U_{i, i+1}}{t_{i+1} - t_i}$$
for the fractional kernel $K(t) = t^{H-1/2} / \Gamma(H+1/2)$ with $H \in \{0.1,-0.1, -0.4\}$.

We can clearly observe the impact of the Hurst index $H$ on the sample paths of the process. If $H > 0$, the process $U$ is absolutely continuous with respect to the Lebesgue measure, and $\widehat{V}$ serves as an approximation of the density process $V$, which satisfies \eqref{eq:V}. In contrast, when $H < 0$, $U$ is no longer absolutely continuous with respect to the Lebesgue measure, and this transition is evident in the sample paths of both $\widehat{U}$ and $\widehat{V}$. As $H \to -0.5$, jump behavior emerges in $\widehat{U}$ and $\widehat{Z}$, illustrating their convergence toward a jump Inverse Gaussian process, as recently established in \cite*{jaber2025hyper}. 

Notably, in all cases, all sample paths of $\widehat{U}$ remain non-decreasing, as proved in Theorem~\ref{T:nonnegative}, and consequently $\widehat{V}$ stays nonnegative.
 
\begin{figure}[h!]
    \centering
    \includegraphics[width=.9\textwidth]{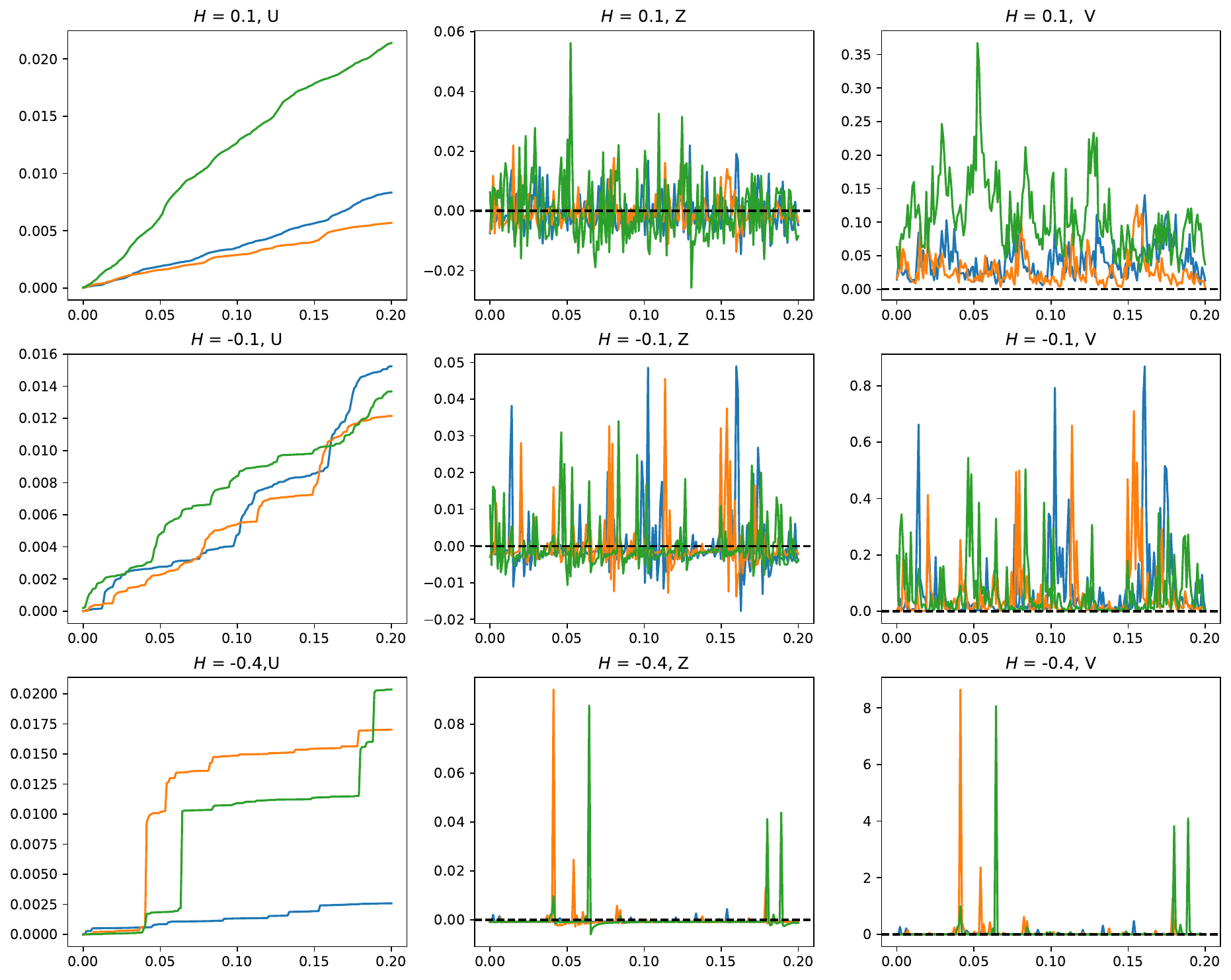} % Path to your PDF file
    \caption{The parameters are $a = 0.1$, $b = -0.3$, $c = 0.2$, $V_0 = 0.04$ 
with a varying Hurst index $H \in \{0.1,-0.1, -0.4\}$ for each row. $T = 0.2$ and $200$  time steps. }
    \label{fig:samplepaths} % Optional: for referencing the figure
\end{figure}

\subsection{Distributional properties of the iVi scheme}\label{S:whyIG}

The Inverse Gaussian distribution is not the only one satisfying \eqref{eq:IGequation}, which naturally raises the question: 
\emph{Why choose the Inverse Gaussian distribution?} In this section, we provide a justification based on distributional properties of the Volterra process  $U$.

To begin, the next proposition highlights striking similarities between the conditional characteristic functions of $\widehat{U}$ and $U$: both exhibit an exponentially affine dependence on $\alpha$ and $G$. More importantly, the Inverse Gaussian distribution emerges naturally as an implicit Euler-type discretization of the Riccati Volterra equations governing the  conditional characteristic function of the increments of  $U$.   In what follows, we use the principal branch for the complex square-root. 

\begin{proposition}\label{P:charcomparison}
   Fix $w\in \mathbb C$ such that $\Re(w)\leq 0$. Fix $i=0,\ldots, n-1$ and $v\geq 0$. Then, 
    \begin{align}
        \mathbb E\left[ \exp\left( w\widehat U_{i,i+1}\right) \mid \alpha_i\right] &= \exp \left( \widehat \psi_{i,i+1} \frac{\alpha_i}{k_0} \right), \label{eq:charhatU}\\
        \mathbb E\left[ \exp\left( wU_{t_i,t_{i+1}}\right) \mid \mathcal F_{t_i} \right] &= \exp\left(  \int_{t_i}^{t_i+1} \left(\frac { c^2\psi^2(t_{i+1}-s)} 2 + b\psi(t_{i+1}-s) + w\right)dG_{t_i}(s) \right), \label{eq:charU}
         \end{align}
         where 
    \begin{align}\label{eq:widehatpsi}
   \widehat \psi_{i,i+1} = \frac{  (1-bk_0) -  \sqrt{(1-bk_0)^2 - 2 w c^2 k_0^2}}{c^2 k_0},
\end{align}
with $k_0$ given by \eqref{eq:kij}   and      $\psi$ is such that $\Re (\psi) \leq 0$ and
solves the Riccati Volterra equation 
         \begin{align}
             \psi(t) &=    \int_0^t K(t-s) \left(w + b \psi(s) +  \frac{c^2}2 \psi^2(s)\right) \, ds, \quad t\geq 0. \label{eq:Ricvariation}
         \end{align}
In particular, $\widehat \psi_{i,i+1}$ is a root of the quadratic polynomial
\begin{align}\label{eq:rootpsi}
    \widehat \psi = k_0 w + k_0 b \widehat \psi + \frac{k_0c^2}{2} \widehat \psi^2.
\end{align}
\end{proposition}

\begin{proof}
We first prove \eqref{eq:charhatU} using the characteristic function of the Inverse Gaussian distribution recalled in \eqref{eq:IGchar}.  Conditional on $\alpha_i$, it follows from \eqref{eq:hatUsample} that $\widehat U_{i,i+1}$  follows an Inverse Gaussian distribution with mean parameter $\mu_i = \frac{\alpha_i}{1-bk_0}$  and shape parameter $\lambda_i = \frac{\alpha_i^2}{c^2k_0^2}$,    which plugged in \eqref{eq:IGchar}  yields \eqref{eq:charhatU}. In addition, it is straightforward to check that $\widehat \psi_{i,i+1}$  solves \eqref{eq:rootpsi}. As for \eqref{eq:charU} and \eqref{eq:Ricvariation}, it follows from the conditional characteristic function \eqref{eq:hestonchar} applied with $v=0$ between $t_i$ and $t_{i+1}$. 
\end{proof}

With the help of Proposition~\ref{P:charcomparison}, the choice of the Inverse Gaussian distribution can now be justified by a simple discretization of the Riccati equation \eqref{eq:Ricvariation} as follows. 
\begin{remark}\label{R:charcomparison}
    The first step is to  discretize \eqref{eq:Ricvariation} using the right endpoint rule by  writing $\Delta_i =t_{i+1} -t_i $:
$$ \psi(\Delta_i) \approx  w \int_0^{\Delta_i}K(s) \, ds +  \int_0^{\Delta_i} K(\Delta_i -s) \, ds \left( b\psi(\Delta_i) +  \frac{c^2}{2}\psi^2(\Delta_i)\right) = w k_0+ bk_0  \psi^2(\Delta_i) + \frac{k_0c^2}{2} \psi^2(\Delta_i).   $$
This yields the quadratic equation \eqref{eq:rootpsi} as approximation for $\psi(\Delta_i)$. The root with a non-positive real part is given precisely by $\widehat \psi_{i,i+1}$ in \eqref{eq:widehatpsi}. 

It follows that the logarithm of the right-hand side of \eqref{eq:charU} can be approximated by 
\begin{align}
    \int_{t_i}^{t_i+1} \left(\frac {c^2 \psi^2(t_{i+1}-s)} 2 + b \psi(t_{i+1}-s) + w\right)dG_{t_i}(s) &\approx  \left(\frac {c^2 \psi^2(\Delta_i)} 2 + b \psi(\Delta_i) + w\right)  \int_{t_i}^{t_i+1}  dG_{t_i}(s) \\
    &\approx    \frac{\widehat \psi_{i,i+1}}{k_0} G_{t_i}(t_{i+1}) \approx  \frac{\widehat \psi_{i,i+1}}{k_0} \alpha_i,
\end{align}
where used the fact that $\widehat \psi_{i,i+1}$ solves  \eqref{eq:rootpsi}, $G_{t_i}(t_i)=0$ and the approximation $G_{t_i}(t_{i+1}) \approx \alpha_i$ as in \eqref{eq:Gapprox}. 
 In other words, the discretization between $t_i$ and $t_{i+1}$ of the Riccati Volterra equation   \eqref{eq:Ricvariation}  that governs the conditional distribution of the integrated process $U_{t_i,t_{i+1}}$ in \eqref{eq:charU} naturally leads to an Inverse Gaussian distribution of the form \eqref{eq:charhatU}:
 $$  \mathbb E\left[ \exp\left( wU_{t_i,t_{i+1}}\right) \mid \mathcal F_{t_i}\right]  \approx \exp \left( \widehat \psi_{i,i+1} \frac{\alpha_i}{k_0} \right) =    \mathbb E\left[ \exp\left( w\widehat U_{t_i,t_{i+1}}\right) \mid \alpha_i \right].$$
 \end{remark}

\subsection{A remark on eliminating the drift}\label{S:Resolvent}

In many applications, we assume $b = 0$ in \eqref{eq:U0}. Without loss of generality, this can always be achieved by modifying the kernel $K$ and the input curve $G_0$ accordingly. Moreover, the discretization of the linear drift term $bU$ introduces a bias, affecting numerical accuracy.

A more accurate scheme can be obtained by explicitly solving for the linear drift term $bU$ in the dynamics of $U$ in \eqref{eq:U0}. This can be done using the {resolvent of the second kind}, which provides a variation of constants formula for linear Volterra equations. The resolvent kernel $R^b$ associated with $bK$ is the unique locally integrable function satisfying  
\begin{equation}
    R^b(t) = bK(t) + \int_0^t bK(t-s) R^b(s) ds, \quad t \geq 0.
\end{equation}
By \cite[Theorems 2.3.1 and 2.3.5]{gripenberg1990volterra} the resolvent exists and  the process $U$ in \eqref{eq:U0} satisfies  
\begin{equation}\label{eq:Uvariation}
    U_{0,t} = \int_0^t\tilde g_0(s)ds  + \frac{c}{b} \int_0^t R^b(t-s) Z_{0,s} ds,
\end{equation}
with $\tilde g_0(t) = g_0(t) + \int_0^tR^b(t-s)g_0(s)ds$ and  the convention that $\frac{R^b}{b} = K$ if $b = 0$. Then, one can implement the iVi scheme of Algorithm \ref{alg:simulation} to \eqref{eq:Uvariation}, i.e.~setting $b=0$ in \eqref{eq:alphai}-\eqref{eq:hatUsample}-\eqref{eq:Zii} and  replacing $g_0(t)$ by $\tilde g_0(t)$ and $K$ by $R^b/b$.  The price to pay is that now \eqref{eq:kij} need to be computed or approximated numerically with the kernel $R^b/b$ instead of $K$. 

\begin{example}
For the constant kernel 
$K\equiv 1$, we have that $R^b(t)=be^{b t}$ and one recovers the iVi scheme developed for the standard square-root process by \cite{abijaber2024simulation}.~For the fractional kernel $K(t)=\gamma t^{H-1/2}$ with $H \in (-1/2,1/2]$, $\gamma>0$, we have that $R^b(t) = b\,\gamma\, \Gamma(H + 1/2)\,  t^{H-1/2}\,E_{H+1/2, H+1/2}\left(b\,\gamma\, \Gamma(H + 1/2)\,t^{H+1/2}\right)$, where $E_{\alpha, \beta}(z)  = \sum_{n \geq 0} \frac{z^n}{\Gamma(\alpha n + \beta)}$ denotes the Mittag–Leffler function. \qed
\end{example}

\section{Convergence of the iVi scheme}\label{S:convergence}
We consider $K \,, g_0 \in L^1_{\rm loc}\left(\mathbb{R}_+\,, \mathbb{R}_+\right)\,$.~Throughout this section, since our aim is to derive the convergence of the scheme as $n \to + \infty\,$, we make the dependence in $n$ explicit and denote the time grid of Algorithm \ref{alg:simulation} by $t^n_i := i \, \frac{T}{n}$ for $i = 0\,,...\,,n$ and $n \geq 1\,$. Moreover, we denote by $\left(\bar K^n, \bar g_0^n\right)$ the objects defined in \eqref{eq:barGbarK}, and by $\left(\widehat U^n_{i,i+1}\,, \widehat Z^n_{i,i+1}\,, k^n_i\,, \alpha^n_i\right)_{i \leq n-1}$ the quantities obtained from Algorithm \ref{alg:simulation}, for $n \geq 1\,$. We assume that all random quantities are defined on the same complete probability space $\left(\Omega, \mathcal F, \mathbb P\right)\,$.

In order for the scheme to be well-defined for any $n \geq 1\,$, we consider the following assumption.
\begin{assumption}
\label{assumption:preserve_nonnegativity}
    For any $n \geq 1\,$, $\bar K^n$ is non-increasing and preserves nonnegativity, and $\bar g_0^n$ is non-decreasing and nonnegative. 
\end{assumption}
According to Theorem \ref{T:nonnegative}, under Assumption \ref{assumption:preserve_nonnegativity}, Algorithm \ref{alg:simulation} is well defined for any $n \geq 1\,$, and we can consider the random variables obtained from it. They allow us to construct the piecewise constant processes
\begin{equation}
    \label{eq:def_process_U}
    U^n_t := \sum_{i = 0}^{\lfloor nt/T \rfloor - 1}\, \widehat U^n_{i,i+1}\,, \quad t \leq T\,, \quad n \geq 1\,,
\end{equation}
\begin{equation}
    \label{eq:def_process_Z}
    Z^n_t := \sum_{i = 0}^{\lfloor nt/T \rfloor -1}\, \widehat Z^n_{i,i+1}\,, \quad t \leq T\,, \quad n \geq 1\,.
\end{equation}
These are càdlàg processes, adapted to the right-continuous filtration
\begin{equation}
\label{eq:def_continuous_filtration}
    \mathcal F^n_t := \sigma\left(\widehat{ \mathcal{F}}^n_t\, \bigcup \mathcal F_0\right)\,, \quad t \leq T\,, \quad n \geq 1\,,
\end{equation}
where
$$
    \widehat{\mathcal{F}}^n_t := \sigma\left(\widehat U^n_{j,j+1} \,,\, 0 \leq j \leq \lfloor nt/T\rfloor - 1 \right)\,, \quad t \leq T\,, \quad n \geq 1\,,
$$
and $\mathcal F_0$ denotes the set of $\mathbb P$-null elements of $\mathcal F\,$, so that $\left(\mathcal F^n_t\right)_{t \leq T}$ satisfies the usual conditions.
In particular, $U^n$ is non-decreasing for any $n \geq 1\,$. In Lemma \ref{lemma:martingality}, we prove that $Z^n$ is a square-integrable $\left( \mathcal F^n_t\right)_{t \leq T}$-martingale, which is an important property that we expect the scheme to preserve. Although the quadratic variation of $Z^n$ is not $U^n$ for finite $n\,$, the same lemma yields the martingality of $\left(Z^n\right)^2 - U^n$ along with the fact that $\mathbb E \left[\left[Z^n\right]_t\right] = \mathbb E \left[U^n_t\right]$ for $t \leq T\,$. Finally, the right or left endpoint approximations we use to derive the scheme \eqref{eq:Uti} preserve the Volterra-type structure, provided the kernel is replaced by a measure-valued one. Introducing the nonnegative discrete measure $K^n := \sum_{i = 0}^{n-1}\, k^n_i\, \delta_{t^n_i}$ on $\left([0\,,T]\,, \mathcal B_{[0\,,T]}\right)\,$, where the $\left(k^n_i\right)_{i \leq n-1}$ are defined in \eqref{eq:kij}, we show in Lemma \ref{lemma:convolution_equation} that 
\begin{equation}
\label{eq:convolution_equation_measure_kernel}
U^n_t = \int_0^{\lfloor \frac{nt}{T} \rfloor \frac{T}{n}}\, g_0(s)\,ds + \int_{[0,t]}\, \left(b\,U^n_{t-s} + c\,Z^n_{t-s}\right)\, K^n(ds)\,, \quad t \leq T\,, \quad n \geq 1\,.
\end{equation}
This reformulation of the scheme plays a key role in the proof of the following result. We denote by $D\left([0\,,T]\right)$ the space of real-valued càdlàg functions on $[0\,,T]$ and we endow it with the Skorokhod $J_1$ topology.
\begin{theorem}
    \label{theorem:weak_convergence}
    Let $K\,, g_0 \in L^1_{\rm loc}\left(\mathbb{R}_+\,, \mathbb{R}_+\right)$ such that Assumption \ref{assumption:preserve_nonnegativity} is satisfied. Then, the sequence $\left(U^n\,, Z^n\right)_{n \geq 1}$ defined in \eqref{eq:def_process_U}-\eqref{eq:def_process_Z} is $J_1\left(\mathbb{R}^2\right)$-tight on $D([0\,,T])^2\,$.~Moreover, any accumulation point $\left(U\,,Z\right)$ satisfies the following:
    \begin{enumerate}[label=(\roman*)]
        \item $U$ is a continuous, non-decreasing process, starting from $0\,$.
        \item $Z$ is a continuous square-integrable martingale with respect to the filtration generated by $\left(U\,,Z\right)$, starting from $0\,$, such that 
        $$
        \langle Z \rangle_t = U_t\,, \quad t \leq T\,.
        $$
        \item The following Volterra equation holds:
        $$
        U_t = \int_0^t\, g_0(s)\,ds + \int_0^t \, K(t-s) \left(b\,U_s+ c\,Z_s\right)\,ds \,, \quad t \leq T\,, \quad \text{a.s.}
        $$
    \end{enumerate}
\end{theorem}
\begin{proof}
    The proof is given in Section \ref{S:ProofConvergence}.
\end{proof}
This proves the weak existence of a couple of càdlàg processes $\left(U\,,Z\right)$ satisfying (\textit{i}), (\textit{ii}) and (\textit{iii}), provided that Assumption \ref{assumption:preserve_nonnegativity} is satisfied, thereby extending the framework considered in \cite{abi2021weak}. However, in order to deduce weak convergence from Theorem \ref{theorem:weak_convergence}, we need weak uniqueness of the limit.~This is the case for the classical kernels, i.e. fractional or exponential, that are completely monotone. This is specified in the following corollary.
\begin{corollary}
    \label{corollary:weak_convergence_monotone}
    If the law of a couple of processes $\left(U\,,Z\right)$ satisfying (\textit{i}), (\textit{ii}) and (\textit{iii}) of Theorem \ref{theorem:weak_convergence} is unique, then we have the weak convergence 
    $$
    \left(U^n\,, Z^n \right) \overset{n \to \infty}{\implies} \left(U\,,Z\right)
    $$
    in $D([0\,,T])^2$ endowed with the topology $J_1\left(\mathbb{R}^2\right)\,$.

    In particular, this is the case if $K$ is completely monotone (see Example \ref{example:completely_monotone}).
\end{corollary}
\begin{proof}
    The first point is direct from tightness and uniqueness of the limiting law.~When $K$ is completely monotone, the weak uniqueness is given by \citet[Theorem~2.2, Example~2.4, Theorem~2.5]{abi2021weak}.
\end{proof}

\section{Numerical illustration}
\label{section:numerical}

For our numerical illustrations, we will consider the shifted fractional kernel 
\begin{align}
K_{H,\epsilon}(t)={\beta (\epsilon + t)^{H-1/2}}, 
\end{align}
for some  $\beta \geq 0$, $\epsilon \geq 0$ and $H\leq 1/2$. Note that if $\epsilon >0$, then the shifted fractional kernel is locally integrable for any $H\leq 1/2$, whereas for $\epsilon=0$, the kernel is singular at $0$ and locally integrable whenever $H\in (-1/2,1/2]$. 

In addition, we will consider an input curve of the form
\begin{align}
   g_0(t) = V_0 + a\int_0^t K_{H,\epsilon}(r) dr = V_0 + a \beta \frac{(\epsilon + t)^{H + 0.5} - \epsilon^{H + 0.5}}{(H + 0.5 )},
\end{align}
for some $V_0,a\geq 0$.

It follows that the first part in \eqref{eq:alphai}  is explicitly given by 
$$ \int_{t_i}^{t_{i+1}} g_0(s) ds = V_0  (t_{i+1} -t_i) +  a \beta \frac{(\epsilon + t_{i+1})^{H + 1.5} -(\epsilon +t_i)^{H + 1.5}}{(H + 1.5 )(H + 0.5 )} - a \beta  (t_{i+1}-t_i) \frac{\epsilon^{H+0.5}}{(H+0.5)}    $$
and that $k_\ell$ in \eqref{eq:kij} is equal to 
$$ k_\ell = \beta \frac{(\epsilon+(\ell +1)\frac{T}{n}))^{H + 0.5} -(\epsilon+\ell \frac{T}{n})^{H + 0.5}  }{(H + 0.5 )}, \quad \quad \ell=0,\ldots, n. $$
We take $\beta =1/ \Gamma(H+0.5)$, where $\Gamma$ is the Gamma function.

We collect the chosen parameter values in Table~\ref{tab:parameter_cases}. Case 1 corresponds to a use case from \cite*{abi2019multifactor, richard2023discrete} and features a nonzero $b$. The remaining cases set $b = 0$, which is common in practice. Case 3 explores an extreme negative value of $H$, while Case 4 corresponds to the shifted fractional kernel. The latter better aligns with market option prices in practice, particularly capturing the concave at-the-money skew  (in log-log plot) observed for equity options \citep{abi2024volatility, guyon2022does}.~The parameter $\rho$ is not used in the Section~\ref{S:U} for the simulation of the process $U$ but will become relevant in the Section~\ref{S:HestonNumerics}, where the scheme is applied to the Volterra Heston model.

\begin{table}[h!]
\centering
\[
\begin{array}{|c|c|c|c|c|c|c|c|}
\hline
\textbf{Case} & a & b & c & \rho & V_0 & H & \epsilon \\ 
\hline
\text{\textbf{Case 1}} & 0.02 & -0.3 & 0.3 & -0.7 & 0.02 & 0.1 & 0 \\ 
\hline
\text{\textbf{Case 2}} & 0.04 & 0 & 0.7 & -0.7 & 0.02 & 0.1 & 0 \\ 
\hline
\text{\textbf{Case 3}} & 0.04 & 0 & 0.7 & -0.7 & 0.02 & -0.3 & 0 \\ 
\hline
\text{\textbf{Case 4}} & 0.04 & 0 & 0.9 & -0.7 & 0.06 & 0 & 1/52 \\ 
\hline
\end{array}
\]
\caption{Parameter values for the four cases.}
\label{tab:parameter_cases}
\end{table}

As a comparison, we also implement a scheme derived from an explicit discretization of \eqref{eq:Ust} between $t_i$ and $t_{i+1}$ using \eqref{eq:Gapprox} with  $\alpha_i$ given by \eqref{eq:alphai} with $\left(\bar U,\bar Z\right)$ instead of $\left(\widehat U, \widehat Z\right)$ defined as follows:  
\begin{align}\label{eq:explicitscheme}
    \bar U_{i,i+1} = (\alpha_i)^+, \quad  \bar Z_{i,i+1} = \sqrt{\bar U_{i,i+1}} \, \xi_i, \quad i = 0, \ldots, n-1,
\end{align}  
where $(x)^+ := \max(x,0)$ for $x \in \mathbb{R}$, and $(\xi_i)_{i=0,\ldots,n-1}$ are independent standard Gaussian variables. Since $\bar U_{i,i+1}$ may become negative, we take its positive part in the definition of $\bar Z_{i,i+1}$.   A similar scheme is studied in \cite*[Equation (9)]{richard2023discrete}, where the authors approximate  
$\int_{t_i}^{t_{i+1}} K(t - t_j) dt \approx (t_{i+1} - t_i) K(t_i - t_j).$  
Our formulation here, which computes the integrals in \eqref{eq:kij} exactly, is expected to be more precise, particularly for singular kernels like the fractional kernel. In any case, \cite*{richard2023discrete} show that such an explicit scheme on the integrated process $U$ outperforms  numerically an explicit Euler scheme built on the density process $V$ of $U$, whenever such a process exists.

\subsection{Numerical illustrations for the  process $U$}\label{S:U}

In this section, we examine the performance of the iVi scheme given by Algorithm~\ref{alg:simulation} for quantities involving the accumulated integrated process
$$U_{0,T} = \int_0^T V_s \, ds = \sum_{i=0}^{n-1} \int_{t_i}^{t_{i+1}} V_s \, ds = \sum_{i=0}^{n-1} U_{t_i,t_{i+1}},$$
which can be naturally approximated using our scheme by
\begin{tcolorbox}[colback=gray!20, colframe=gray!80, sharp corners]
$$\widehat{U}_{0,T} := \sum_{i=0}^{n-1} \widehat{U}_{i,{i+1}}.$$
\end{tcolorbox}
We focus on the Laplace transform $\mathbb{E}[e^{-U_1}]$. The explicit reference value  is computed using  the Volterra Riccati equation as recalled  in \eqref{eq:hestonchar}.

\begin{figure}[h!]
    \centering
\includegraphics[width=1.\textwidth]{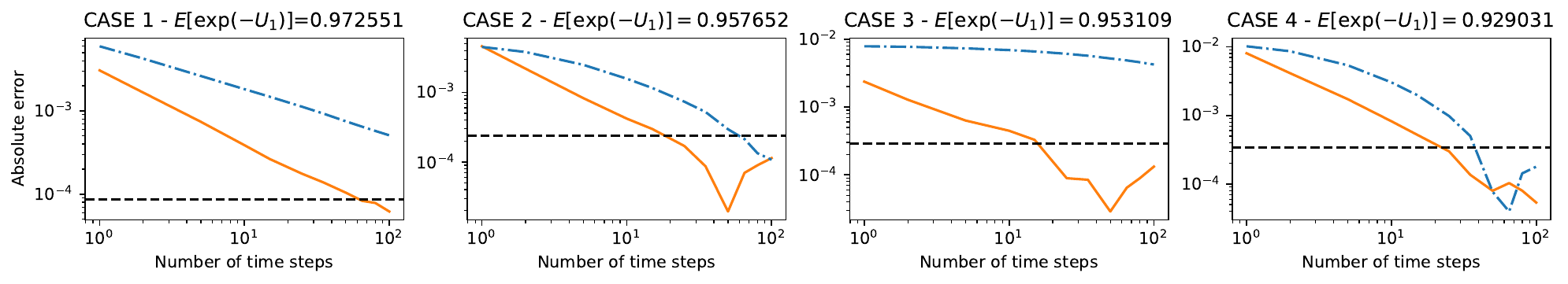}
    \caption{Errors on  Laplace transform of $\hat U_T$ in terms of the number of time steps  for  the four cases with  $T = 1$ and 1 million sample paths. Plain line for the iVi scheme, blue dotted line for the benchmark.}
    \label{fig:U} 
\end{figure}

Figure~\ref{fig:U} shows the absolute error between the schemes and reference values, varying the number of time steps on a uniform grid from 1 to 100 in log-log scale. Simulations use 1 million sample paths. The horizontal black dotted line represents three standard deviations of the Monte Carlo estimator, below which comparisons can be considered non-significant. 

We observe that the implicit iVi scheme of Algorithm~\ref{alg:simulation} consistently outperforms its explicit counterpart given by \eqref{eq:explicitscheme} across all parameter sets, with less than 100 time steps, particularly in the extreme regime of Case 3 with $H = -0.3$. Additionally, in Case 1, when $b \neq 0$, convergence appears to be slower compared to cases where $b = 0$, which is consistent with Section~\ref{S:Resolvent}.

%To sum up, for quantities on the integrated process $U$, 
%the iVi scheme converges and competes favorably with the QE scheme,  particularly  in challenging regimes with high volatility-of-volatility and strong mean reversion.

\subsection{Numerical illustrations  for the Volterra Heston model}\label{S:HestonNumerics}
In this section, we test our iVi scheme on   the class of Volterra Heston models where the stock price $S$ is given by 
\begin{align}\label{eq:HestonS}
    dS_t = S_t  \left(\rho dW_{U_{0,t}} + \sqrt{1-\rho^2}dW_{U_{0,t}}^{\perp}\right), \quad S_0 >0,
\end{align}
where $U$ is given by \eqref{eq:U0}, $\rho \in [-1,1]$ and $W^{\perp}$ is a standard Brownian motion independent of $W$.  We refer to \cite[Section 7]{abi2021weak} for $L^1_{\rm loc}$ kernels and \cite*[Section 7]{abijaber2019affine} for $L^2_{\rm loc}$ kernels. In the context of the fractional kernel, for $H \in (-1/2,0]$ we recover hyper-rough models in the sense of  \cite{jusselin2020no} and for $H \in (0,1/2]$, the rough Heston model of \cite{el2019characteristic}.

In order to simulate $S$ it suffices to observe that 
$$ \log S_{t_{i+1}} = \log S_{t_i}  - \frac 1 2 U_{t_{i},t_{i+1}} + \rho Z_{t_i, t_{i+1}}  + \sqrt{1-\rho^2} \int_{t_i}^{t_{i+1}} dW_{U_{0,s}}^{\perp}, $$
and that conditional on $U_{t_i,t_{i+1}}$, $\int_{t_i}^{t_{i+1}} dW_{U_{0,s}}^{\perp} \sim \mathcal N(0, U_{t_i, t_{i+1}})$, for $i=0,\ldots, n-1$.

\begin{tcolorbox}[colback=gray!20, colframe=gray!80, sharp corners]
We can therefore simulate $(\log \widehat S_i)_{i=0,\ldots,n}$
using the outputs $\left(\widehat U, \widehat Z\right) $ of Algorithm \eqref{alg:simulation} using:
\begin{align}
\log \widehat S_0 &= \log S_0, \\
    \log \widehat S_{i+1} &= \log \widehat S_{i}  - \frac 1 2 \widehat U_{{i},{i+1}} + \rho \widehat Z_{i, {i+1}}  + \sqrt{1-\rho^2} \sqrt{ \widehat U_{{i},{i+1}}} N_i, \quad i=0,\ldots, n-1,  \label{eq:Sii}
\end{align}
where $(N_i)_{i=0,\ldots, n-1}$ are i.i.d.~standard Gaussian random variables. 
\end{tcolorbox}

Clearly, the update rule \eqref{eq:Sii} at the $i$-th step $i$ can be incorporated in Algorithm~\ref{alg:simulation} right after  \eqref{eq:Zii}.~Also, the $(N_i)_{i=0,\ldots, n-1}$ need to be taken independent  from the Gaussian and Uniform random variables used for the sampling of the Inverse Gaussian distribution in Algorithm~\ref{alg:IG_sampling}.

The following remark rewrites the model in a more conventional form when $K$ is locally square-integrable.
\begin{remark}\label{R:hyper}
	Assume that $K \in L^2_{\rm loc}$ (e.g. the fractional kernel from \eqref{eq:fractionalkernel} with $H > 0$). It follows from  \citet[Lemma~2.1]{abi2021weak} that  $U_t=\int_0^t V_s ds$ where $(S,V)$ is a Volterra Heston model in the terminology of \citet[Section 7]{abijaber2019affine}; (rough Heston model of \citet{el2019characteristic} for the fractional kernel) satisfying 
	\begin{align*} 
	d\log S_t &= -\frac 1 2 V_tdt + \sqrt{V_t}d\widetilde B_{t},  \quad S_0>0,\\
	V_t  &= g_0(t) + \int_0^t K(t-s)\left(b\,V_s\,ds +  c\,\sqrt{V_s}\, d\widetilde W_s\right),
	\end{align*}	
	for some Brownian motions $\widetilde B$ and $\widetilde W$ with correlation $\rho$ obtained from standard martingale representation theorems on a possible extension of the probability space, see for instance \citet[Proposition V.3.8]{revuz2013continuous}. For the fractional kernel with $H \in (0,1/2)$, the sample paths of $V$ are  H\"older continuous of any order strictly less than $H$ and the process $V$ is said to be `rough'.
\end{remark}

If $K$ is no longer in $L^2_{\rm loc}$ (e.g. the fractional kernel from \eqref{eq:fractionalkernel} with $H\leq0$), not only Fubini's interchange  breaks down, but it can also be shown that $U$ is nowhere differentiable almost surely, see \citet[Proposition 4.6]{jusselin2020no}. In this case, one cannot really make sense of the spot variance $V$ and is stuck with the `integrated variance' formulation \eqref{eq:U0}, justifying the appellation ‘hyper--rough' for such equations.

For our numerical experiment, we will consider call options on $S$ for  the four cases of Table~\ref{tab:parameter_cases} with maturity $T=1$. Reference values are computed using Fourier inversion techniques on the characteristic function of the log-price which is  known  in the Volterra Heston model, see \eqref{eq:hestonchar}.  Simulations use 1 million sample paths and we also benchmark against the explicit scheme  \eqref{eq:explicitscheme}.  

Figure~\ref{fig:call} displays the absolute error between the schemes and reference values, varying the number of times steps from 1 to 100 in log-log scale for  At-the-Money (ATM) options.  We can  again  observe   a significantly faster convergence of the implicit iVi scheme compared to the explicit scheme in all for cases, with less than 100 time steps, and in particular for the extreme case 3 with $H=-0.3$.

\begin{figure}[h!]
    \centering
\includegraphics[width=.8\textwidth]{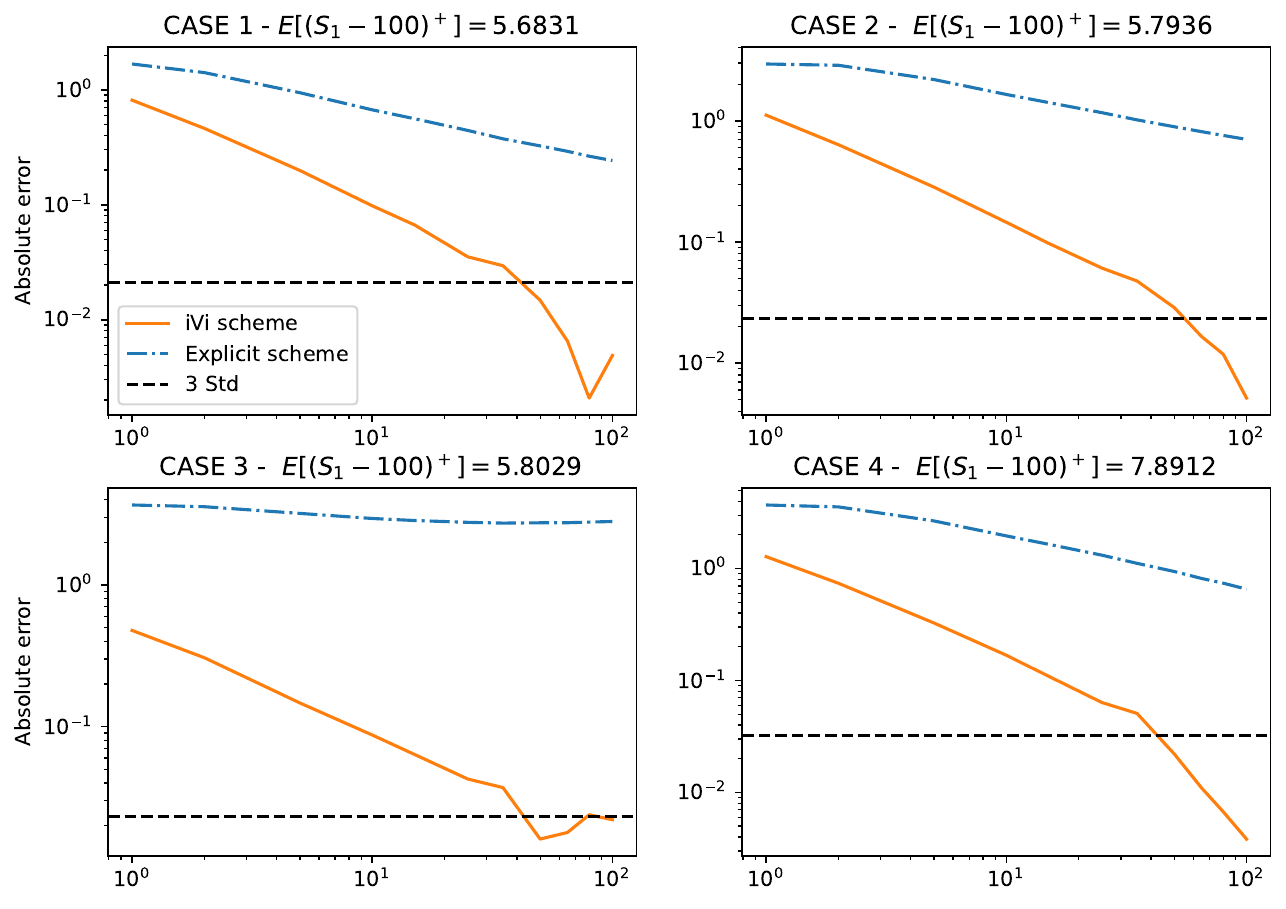} \\
    \caption{ATM Call options  on $S$:  error in prices in  terms of number of time steps for the four cases  with 1 million sample paths. Plain line for the iVi scheme, blue dotted line for the benchmark.}
    \label{fig:call} % Optional: for referencing the figure
\end{figure}

To illustrate the impact of the Hurst index $H$ on the behavior of the scheme, we set the parameters as in Case 2 in Table~\ref{tab:parameter_cases}, except for the parameter $H$, which we let vary between 0.1 and -0.49. For each $H$, we plot in Figure~\ref{fig:surface} the full slice of implied volatilities with maturity $T = 1$ for time steps varying between 1 and 32.  Remarkably, with 32 time steps, the slices of the iVi scheme and the reference values become nearly indistinguishable for all values of $H$. Even more strikingly, the convergence of the iVi scheme appears to be faster as $H$ decreases, which is typically the opposite of what is observed in existing schemes. As $H \to -0.5$, even with a single time step, the scheme closely approximates the implied volatility slice, in line with the limiting behavior derived in \cite*{jaber2025hyper}.

\section{Proof of Theorem~\ref{T:nonnegative}}\label{S:ProofMainWellDefined}

In this section, we prove Theorem~\ref{T:nonnegative}. 

As already mentioned, the idea is to make the problem  continuous in time to bypass cumbersome discrete recursions. For this, we define new random variables $\left(\bar{U}_{i,i+1}\right)_{i \leq n-1}\,$, $\left(\bar{Z}_{i,i+1}\right)$ and $\left(\bar{\alpha}_i\right)_{i \leq n-1}$ that are constructed recursively using Algorithm \ref{alg:simulation}, except for the following point: 
\begin{center}
    For any $i \leq n-1\,$, $\bar{U}_{i,i+1}$ is constructed using \eqref{eq:hatUsample} if $\bar{\alpha}_i \geq 0\,$, and $\bar{U}_{i,i+1} = 0$ otherwise.
\end{center}
Thus, the sequence is well-defined.
Note that at step $i \leq n-1\,$, $\bar{Z}_{i,i+1}$ is still defined as in \eqref{eq:Zii}, i.e. 
$$
\bar{Z}_{i,i+1} := \frac{1}{ck_0}\, \left(\left(1 - bk_0\right)\, \bar{U}_{i,i+1} - \bar{\alpha}_i\right)\,,
$$
and the sequence $\left(\bar{\alpha}_i\right)_{i \leq n-1}$ is defined recursively with $\bar{\alpha}_0 = \alpha_0$ and the update \eqref{eq:alphai} with $\left(\widehat{U}_{i,i+1}\right)_{i \leq n-1}$ and $\left(\widehat{Z}_{i,i+1}\right)_{i \leq n-1}$ replaced by $\left(\bar{U}_{i,i+1}\right)_{i \leq n-1}$ and $\left(\bar{Z}_{i,i+1}\right)_{i \leq n-1}\,$.
\begin{remark}
    \label{remark:alpha_tilde_positive}
    If we manage to prove that for any $i \leq n-1\,$, $\bar{\alpha}_i \geq 0\,$, then it is clear that 
    $$
    \alpha_i = \bar{\alpha}_i \,, \quad i \leq n-1\,,
    $$
     and the $\left(\alpha_i\right)_{i \leq n-1}$ are all nonnegative, showing that Algorithm \ref{alg:simulation} is well-defined.
\end{remark}
In that purpose, we define the càdlàg process $A$ recursively as follows:
\begin{itemize}
    \item 
 For $i=0$, on $[0,t_1)$, we  define 
\begin{align}\label{eq:defA0}
    A(t) = \bar g_0(t_0), \quad t \in [0,t_1).
\end{align}
\item For $i\geq 1$, having constructed $A$ on $[0,t_i)$, we set $A$ on $[t_i, t_{i+1})$ by:
\begin{align}\label{eq:defAi}
    A(t) = \bar g_0(t_i) + \sum_{j=1}^i \frac{\bar K(t - t_j)}{\bar K(0)} \left( \bar U_{j-1, j} - A(t_{j} -) \right), \quad t \in [t_{i},t_{i+1}),
\end{align}
\iffalse where $\widehat U_{j-1, j}$ are constructed using \eqref{eq:hatUsample} if $\alpha_{j-1} \geq 0$, and $\widehat U_{j-1, j} = 0$ in case $\alpha_{j-1} < 0$, so that $A$ is well-defined.\fi and we finally set $A(t_n)=A(T)=A(T-)$.
\end{itemize}

In order to prove Theorem~\ref{T:nonnegative}, it suffices to prove that 
$$\bar \alpha_i\geq 0, \quad i=0, \ldots, n-1\,, $$ as mentioned in Remark \ref{remark:alpha_tilde_positive}.
We will achieve this by showing that  $A(t_{i+1}-) = \bar \alpha_i$, for all $i=0, \ldots, n-1$,  and 
$ A(t) \geq 0  $, for all $t\leq T$, successively in Lemmas~\ref{L:A1} and \ref{L:A2}.

\begin{lemma}\label{L:A1} 
    Let $K, g_0 \in L^1_{\rm loc}(\R_+, \R)$. Then, 
    \begin{align}
        A(t_i) &= \bar g_0(t_i) - \bar g_0(t_{i-1}) + \bar U_{i-1,i}, \label{eq:LA11} \\
        A(t_{i+1}-) &= \bar \alpha_i,  \quad \quad \quad \quad \quad \quad \quad \quad \quad \quad \quad  i= 0, \ldots, n-1. \label{eq:LA12}  
    \end{align}
\end{lemma}

\begin{proof}
    For $i=0$, we readily have from the definition of $A$ in \eqref{eq:defA0} that 
    \begin{align}
        A(t_0) = A(t_1 - ) = \bar g_0(t_0) = \int_0^{\frac T n} g_0(s) \, ds = \alpha_0, 
    \end{align}
    which yields \eqref{eq:LA11}-\eqref{eq:LA12} for $i=0$.  
    
   $\bullet$ We first prove \eqref{eq:LA11} for $i\geq 1$. From the definition of $A$ on $[t_i,t_{i+1})$ in \eqref{eq:defAi}, we obtain for $t=t_i$ that 
    \begin{align}
        A(t_i) =  \bar g_0(t_i) + \bar U_{i-1,i} - A(t_i -) +  \sum_{j=1}^{i-1} \frac{\bar K(t_i - t_j)}{\bar K(0)} \left( \bar U_{j-1, j} - A(t_{j} -) \right).
    \end{align}
   Now, using the definition of $A$ on $[t_{i-1}, t_i)$, and evaluating at $t=t_i-$, we get that 
      \begin{align}
        A(t_i-) =  \bar g_0(t_{i-1}) +  \sum_{j=1}^{i-1} \frac{\bar K(t_i - t_j)}{\bar K(0)} \left( \bar U_{j-1, j} - A(t_{j} -) \right),
    \end{align}
where we used the fact that $\bar K(t_i - t_j) = \bar K( (t_i -) - t_j)$, by continuity of the kernel $\bar K$, {which is straightforward using the dominated convergence theorem since $K \in L^1([0\,,T])\,$}. Combining the the two identities above leads to \eqref{eq:LA11}.

$\bullet$ In order to argue  \eqref{eq:LA12} for $i\geq 1$, we proceed by induction. We fix $i\geq 1$ and assume that $A(t_{j+1}-)=\bar \alpha_j$ for $j=0,\ldots, i$, we want to prove the equality for $(i+1)$. Using the definition of $A$, the induction assumption  and a change of variables, we obtain that 
\begin{align}
    A(t_{i+1}-) &= \bar g_0(t_i) +  \sum_{j=1}^{i} \frac{\bar K(t_{i+1} - t_j)}{\bar K(0)} \left( \bar U_{j-1, j} - A(t_{j} -) \right) \\
    &= \bar g_0(t_i) +  \sum_{j=0}^{i-1} \frac{\bar K(t_{i+1} - t_{j+1})}{\bar K(0)} \left( \bar U_{j, j+1} - A(t_{j+1} -) \right) \\
    &= \bar g_0(t_i) +  \sum_{j=0}^{i-1} \frac{k_{i-j}}{k_0} \left( \bar U_{j, j+1} - \bar \alpha_{j} \right)  \\
    &= \bar \alpha_i
\end{align}
using the definition of $\bar{\alpha}$ and $\bar U\,,$
along with $\bar K(t_{i+1} - t_{j+1}) = k_{i-j}$ where $k_{i-j}$ is given by \eqref{eq:kij}.
\end{proof}

Before stating Lemma~\ref{L:A2}, we recall a key property of non-increasing kernels that preserve nonnegativity when adding an initial term $x_0 \geq 0$. 

\begin{proposition}\label{P:pp}
Let \( K \colon \mathbb{R}_+ \to \mathbb{R}_+ \) be a non-increasing kernel with \( K(0) > 0 \) that preserves nonnegativity. Let \( i \geq 1 \), \( 0 \leq t_1 < \cdots < t_i \), and \( x_0, \ldots, x_i \in \mathbb{R} \) be such that \( x_0 \geq 0 \) and
\[
 x_0 + \sum_{j=1}^\ell  K(t_\ell - t_{j})  x_{j} \geq 0, \quad  \ell= 1, \ldots, i.
\]
Then, we have
\[
x_0 + \sum_{j : t_j \leq t}K(t - t_j)  x_j  \geq 0, \quad  t \geq 0.
\]
\end{proposition}

\begin{proof}
    See \cite[Proposition 2.8]{alfonsi2025nonnegativity}.
\end{proof}

\begin{lemma}\label{L:A2} Let $K, g_0 \in L^1_{\rm loc}(\R_+, \R)$.~Assume that  $\bar g_0:[0,T]\to \R_+$ is non-decreasing and nonnegative  and that $\bar K:\R_+ \to \R_+$ is non-increasing and preserves nonnegativity.  Then, 
$$ A(t) \geq 0, \quad t \in [0,T].$$
\end{lemma}

\begin{proof} We proceed by induction on $i$ with the property that $A$ is nonnegative on  $[t_i,t_{i+1}]$. For $i=0$,  it follows from the definition of $A$ in \eqref{eq:defA0} that for  $t \in [0,t_1)$, $A(t) =  \bar g_0(t_0) \geq 0$, since $\bar g_0 \geq 0$. At $t=t_1,$ we have from \eqref{eq:LA11} that 
$$ A(t_1) = \bar g_0(t_1) - \bar g_0(t_0) + \bar U_{0,1} \geq 0, $$
since $\bar g_0$ is assumed to be non-decreasing and $\bar U_{0,1} \geq 0$ by construction.
For $i \geq 1$, we assume that $A$ is nonnegative on $[0,t_1], [t_1,t_2], \ldots [t_{i-1}, t_i]$ and we want to show that it is nonnegative on $[t_i,t_{i+1}]$. First, at $t=t_{i+1}$, an application of \eqref{eq:LA11}  yields similarly to above that $A(t_{i+1}) \geq 0$. It remains to argue that $A$ is nonnegative on $[t_i,t_{i+1})$ using Proposition~\ref{P:pp}. Set  $x_j= (\bar U_{j-1,j} - A(t_j-))/\bar K(0)$, for $j=1,\ldots, i$. Then, for each $\ell =1, \ldots, i$, we have thanks to the induction assumption  and the definition of $A$ in \eqref{eq:defAi} that
\begin{align}
    0 \leq A(t_\ell) = \bar g_0(t_l) + \sum_{j=1}^\ell \bar K(t_\ell - t_j) x_j.
\end{align}
Now since $\bar g_0$ is non-decreasing we obtain that 
$$ 0 \leq  \bar g_0(t_i) + \sum_{j=1}^\ell \bar K(t_\ell - t_j) x_j, \quad \ell=1, \ldots, i. $$
An application of Proposition~\ref{P:pp}  with the kernel $\bar K$ yields that 
$$ 0 \leq \bar g_0(t_i) + \sum_{j=1}^i \bar  K(t-t_j) x_j = A(t), \quad t \in [t_i,t_{i+1}),$$
and ends the proof. 
\end{proof}

The proof of Theorem~\ref{T:nonnegative} readily follows. 

\begin{proof}[Proof of Theorem~\ref{T:nonnegative}] Fix $i\geq 0$. It follows from \eqref{eq:LA12}, that $\bar \alpha_i  = A(t_{i+1}-)$, which is nonnegative thanks to  Lemma~\ref{L:A2}. This ends the proof, in virtue of Remark \ref{remark:alpha_tilde_positive}.
\end{proof}

\section{Proof of Theorem \ref{theorem:weak_convergence}}
\label{S:ProofConvergence}
We prove Theorem \ref{theorem:weak_convergence} in four steps. First, we derive estimates on the random variables $\left(\alpha_i^n\right)_{i \leq n-1, n \geq 1}\,$, $\left(\widehat{U}^n_{i,i+1}\right)_{i \leq n-1, n \geq 1}$ and $\left(\widehat{Z}^n_{i,i+1}\right)_{i \leq n-1, n \geq 1}\,$. These then allow us to control the processes $\left(U^n\right)_{n \geq 1}$ and $\left(Z^n\right)_{n \geq 1}\,$, from which we deduce the tightness of $\left(U^n\,, Z^n\right)_{n \geq 1}\,$. Finally, we characterize the limiting processes.

\subsection{Estimates on the sequences of random variables}
For $n \geq 1$ and $i \leq n-1\,$, the random variable $\alpha^n_i$ is $\mathcal{F}^n_{t^n_i}$-measurable, and $\widehat U^n_{i,i+1}$ as well as $\widehat{Z}^n_{i,i+1}$ are $\mathcal{F}^n_{t^n_{i+1}}$-measurable (see \eqref{eq:def_continuous_filtration} for a definition of the filtration). Moreover, conditional on $\alpha^n_i\,$,
\begin{equation}
\label{eq:conditional_law_U}
\widehat U^n_{i,i+1} \sim IG\left( \frac{\alpha_i^n}{1 - bk_0^n}\,, \left(\frac{\alpha_i^n}{ck_0^n}\right)^2\right)
\end{equation}
in the sense of Appendix \ref{A:IG}. Leveraging the properties of the Inverse Gaussian distribution, we obtain the following estimates.
\begin{lemma}
    \label{lemma:properties_sequence}
    Let $n \geq 1$ and $0 \leq i \leq n-1\,$, then the random variables $\widehat{U}^n_{i,i+1}\,$, $\widehat{Z}^n_{i,i+1}$ and $\alpha_i^n$ admit moments of order 4, and we have:
    \begin{enumerate}[label=(\roman*)]
        \item $\mathbb{E}\left[\widehat{Z}^n_{i,i+1}\right] = \mathbb{E}\left[ \widehat{Z}^n_{i,i+1}\, |\, \mathcal{F}^n_{t^n_i}\right] = 0$
        \item $\mathbb{E}\left[\left(\widehat{Z}^n_{i,i+1} \right)^2\,|\, \mathcal{F}^n_{t^n_i}\right] = \mathbb{E}\left[\widehat{U}^n_{i,i+1}\,|\,\mathcal{F}^n_{t^n_i} \right]$
        \item $\mathbb{E}\left[\left(\widehat{Z}^n_{i,i+1}\right)^4\right] \leq 5\, \mathbb{E}\left[\left(\widehat{Z}^n_{i,i+1}\right)^2\,\widehat{U}^n_{i,i+1}\right]$
    \end{enumerate}
\end{lemma}
\begin{proof}
    We fix $n \geq 1$ and first prove that our random variables admit a finite fourth moment by induction on $i \leq n-1\,$. We recall that we defined $t^n_j := j\,\frac{T}{n}$ for $j = 0\,,...\,,n$ in Section \ref{S:convergence}. Firstly, $\alpha_0^n = \int_0^{t^n_1}\,g_0(s)\,ds$ is deterministic by definition. Moreover,
    $$
    \widehat{U}^n_{0,1} \sim IG\left(\frac{\alpha^n_0}{1 - b\,k_0^n}\,, \left( \frac{\alpha_0^n}{ck_0^n}\right)^2\right)
    $$
    and
    $$
    \widehat{Z}^n_{0,1} = \frac{1}{ck_0^n}\, \left(\left(1 - bk_0^n \right)\, \widehat{U}^n_{0,1} - \alpha_0^n \right)
    $$
    so that $\widehat{U}_{0,1}^n$ and $\widehat{Z}^n_{0,1}$ have moments of all orders.

    Fix $i \leq n-1$ and suppose $\left(\alpha_j^n\right)_{j < i}$ as well as $\left(\widehat{U}^n_{j,j+1}\right)_{j < i}$ and $\left(\widehat{Z}^n_{j,j+1}\right)_{j< i}$ all admit moments of order 4. Since 
    $$
    \alpha_{i}^n = \int_{t_{i}^n}^{t_{i+1}^n}\,g_0(s)\,ds + \sum_{j = 0}^{i-1} \, k^n_{i-j}\,\left(b\,\widehat{U}^n_{j,j+1} + c\,\widehat{Z}^n_{j,j+1} \right)
    $$
    is a finite sum of $L^4$ random variables, it has a finite moment of order $4\,$. We denote by $\mu_{\alpha^n_i} := \mathbb{P} \circ \left(\alpha^n_i\right)^{-1}$ the law of $\alpha^n_i\,$, and for any $a \geq 0\,$, we denote by $f_a$ the density of the Inverse Gaussian distribution with parameters $\left( \frac{a}{1 - bk_0^n}\,, \left(\frac{a}{ck_0^n} \right)^2\right)$ (see Appendix \ref{A:IG}). Then, the probability density of $\widehat U^n_{i,i+1}$ is given by 
    $$
    x \mapsto \int_{\mathbb{R_+}}\, f_a(x)\, \mu_{\alpha^n_i}(da)\,, \quad x \geq 0\,.
    $$
    By Fubini's theorem, if 
    $$
    \int_{\mathbb{R}_+}\, \int_{\mathbb{R}_+}\,x^4\,f_a(x)\,dx\,\mu_{\alpha_i^n}(da) < +\infty\,,
    $$
    then $\mathbb{E}\left[\left(\widehat U^n_{i,i+1}\right)^4\right] < + \infty\,$. Using \eqref{eq:IGmean}, we obtain 
    $$
    \int_{\mathbb{R_+}}\, x^4\,f_a(x)\,dx = 15\, \frac{\left(ck_0^n\right)^6}{(1 - bk_0^n)^7}\, a + 15 \frac{\left(ck_0^n\right)^4}{\left(1 - bk_0^n\right)^6}\, a^2 + 6\, \frac{\left(ck_0^n\right)^2}{\left(1 - bk_0^n\right)^5}\,a^3 + \left(\frac{a}{1 - bk_0^n}\right)^4\,, \quad a \geq 0\,,
    $$
    and since $\mathbb{E}\left[\left(\alpha_i^n\right)^4\right] < + \infty\,$, we conclude that $\widehat U^n_{i,i+1}$ has a finite moment of order $4\,$.
    \iffalse
    Moreover, let $\xi \sim \mathcal{N}(0,1)$ and $\eta \sim U([0,1])$ independent of each other and of $\mathcal{F}^n_{t^n_i}\,$, then 
    $$
    \hat{U}^n_{i,i+1} \sim X\,\mathbbm{1}_{\eta \leq \frac{\mu}{\mu + X}} + \frac{\mu^2}{X} \mathbbm{1}_{\eta > \frac{\mu}{\mu + X}}
    $$
    where 
    $$
    X := \mu + \frac{\mu^2 \,\xi^2}{2\lambda} - \frac{\mu}{2\lambda}\, \sqrt{4\mu \lambda\,\xi^2+ \mu^2\,\xi^4}
    $$
    with
    $$
    \mu := \frac{\alpha_i^n}{1 - bk_0^n}\,, \quad \lambda := \left(\frac{\alpha_i^n}{ck_0^n} \right)^2 \,
    $$
    from Algorithm \ref{alg:IG_sampling}. We first remark that 
    $$
    0 \leq X \leq \frac{\alpha_i^n}{1 - bk_0^n} + \frac{1}{2}\, \left(\frac{ck_0^n\, \xi}{1 - bk_0^n}\right)^2
    $$
    so that $\mathbb{E}\left[\left|X\right|^4\right] < + \infty\,$. Furthermore, 
    \begin{align*}
        \frac{\mu^2}{X} &= \mu^2\, \frac{\mu + \frac{\mu^2\,\xi^2}{2\lambda} + \frac{\mu}{2\lambda}\, \sqrt{4\mu\lambda\,\xi^2 + \mu^2 \, \xi^4}}{\left(\mu + \frac{\mu^2\, \xi^2}{2\lambda}\right)^2 - \frac{\mu^3\,\xi^2}{\lambda} - \frac{\mu^4\, \xi^4}{4\lambda^2}} \\
        &= \mu^2 \, \frac{\mu + \frac{\mu^2\,\xi^2}{2\lambda} + \frac{\mu}{2\lambda}\, \sqrt{\left(2\lambda + \mu\, \xi^2 \right)^2 - 4\lambda^2}}{\mu^2} \\
        &\leq 2\, \left(\mu + \frac{\mu^2\, \xi^2}{2 \lambda} \right) = \frac{2\,\alpha_i^n}{1 - bk_0^n} + \left(\frac{ck_0^n\, \xi}{1 - bk_0^n} \right)^2 \,,
    \end{align*}
    and $\mathbb{E}\left[\left|\mu^2\,/\,X \right|^4\right] < + \infty\,$. We thus proved that $\widehat{U}^n_{i,i+1}$ is $L^4$, and so is \fi
    By linear combination, this is also the case for
    $$
    \widehat{Z}^n_{i,i+1} = \frac{1}{ck_0^n}\, \left((1 - bk_0^n)\,\widehat{U}^n_{i,i+1} - \alpha_i^n \right)\,.
    $$
    The proof by induction is thus complete.

    We can now easily prove (\textit{i}), (\textit{ii}) and (\textit{iii}) of Lemma \ref{lemma:properties_sequence} using the properties of the Inverse Gaussian distribution. In virtue of \eqref{eq:conditional_law_U} and since $\alpha^n_i$ is $\mathcal F^n_{t^n_i}$-measurable, we have $\mathbb{E}\left[\widehat U^n_{i,i+1}\,|\, \mathcal F^n_{t^n_i}\right] = \mathbb E \left[\widehat U^n_{i,i+1}\,|\, \alpha^n_i\right] = \frac{\alpha^n_i}{1 - bk_0^n}$ for $i = 0\,, ...\,, n-1\,$. Thus, we obtain
    $$
\mathbb{E}\left[\widehat{Z}^n_{i,i+1}\,|\,\mathcal{F}^n_{t^n_i}\right] = \frac{1 - bk_0^n}{ck_0^n}\, \mathbb{E}\left[\widehat{U}^n_{i,i+1} - \mathbb{E}\left[\widehat{U}^n_{i,i+1}\,|\,\mathcal{F}^n_{t^n_i}\right]\,|\,\mathcal{F}^n_{t^n_i} \right] = 0 \,, \quad i \leq n -1\,,
    $$
    so that by the tower property of conditional expectation,
    $$
    \mathbb{E}\left[\widehat{Z}^n_{i,i+1}\right] = \mathbb{E}\left[\mathbb{E}\left[\widehat{Z}^n_{i,i+1}\,|\,\mathcal{F}^n_{t^n_i}\right]\right] = 0\,, \quad i \leq n-1\,,
    $$
    which proves (\textit{i}). Similarly, if $Y \sim IG(\mu, \lambda)$ for $\mu\,,\lambda > 0\,$, we have
    $$
    \mathbb{E}\left[\left(Y - \mu\right)^2\right] = \frac{\mu^3}{\lambda}\,,
    $$
    thanks to \eqref{eq:IGmean},
    which gives
    \begin{align*}
    \mathbb{E}\left[\left( \widehat{Z}^n_{i,i+1} \right)^2\, |\, \mathcal{F}^n_{t^n_i}\right] &= \left(\frac{1 - bk_0^n}{ck_0^n}\right)^2\,\mathbb{E}\left[\left(\widehat{U}^n_{i,i+1} - \mathbb{E}\left[\widehat{U}^n_{i,i+1}\,|\,\mathcal{F}^n_{t^n_i}\right]\right)^2\,|\,\mathcal{F}^n_{t^n_i}\right] \\
    &= \left(\frac{1 - bk_0^n}{ck_0^n}\right)^2\,\left(\frac{ck_0^n}{1 - bk_0^n}\right)^2\,\frac{\alpha_i^n}{1 - b k_0^n}\\
    &= \frac{\alpha_i^n}{1 - bk_0^n}\\
    &= \mathbb{E}\left[ \widehat{U}^n_{i,i+1}\,|\,\mathcal{F}^n_{t^n_i}\right]\,, \quad i \leq n-1\,,
    \end{align*}
    and proves (\textit{ii}).
    Finally, using the same notation $Y \sim IG(\mu,\lambda)\,$, \eqref{eq:IGmean} also gives
    $$
    \mathbb{E}\left[\left(Y - \mu\right)^2\, Y\right] = 3 \, \frac{\mu^5}{\lambda^2} + \frac{\mu^4}{\lambda}\,,
    $$
    and
    $$
    \mathbb{E}\left[(Y - \mu)^4\right] = 15\, \frac{\mu^7}{\lambda^3} + 3\, \frac{\mu^6}{\lambda^2} \leq 5\, \frac{\mu^2}{\lambda}\,\left(3\,\frac{\mu^5}{\lambda^2} + \frac{\mu^4}{\lambda}\right) = 5 \, \frac{\mu^2}{\lambda}\, \mathbb{E}\left[(Y- \mu)^2\,Y\right]\,. 
    $$
    Thus,
    \begin{align*}   \mathbb{E}\left[\left(\widehat{Z}^n_{i,i+1}\right)^4\,|\, \mathcal{F}^n_{t^n_i}\right] &= \left(\frac{1 - bk_0^n}{ck_0^n}\right)^4\, \mathbb{E}\left[\left(\widehat{U}^n_{i,i+1} - \mathbb{E}\left[\widehat{U}^n_{i,i+1}\,|\,\mathcal{F}^n_{t^n_i}\right]\right)^4\, |\, \mathcal{F}^n_{t^n_i}\right] \\
        &\leq 5\,\left(\frac{1 - bk_0^n}{ck_0^n} \right)^2\,\mathbb{E}\left[\left(\widehat{U}^n_{i,i+1} - \mathbb{E}\left[\widehat{U}^n_{i,i+1}\,|\,\mathcal{F}^n_{t^n_i} \right]\right)^2\, \widehat{U}^n_{i,i+1}\,|\,\mathcal{F}^n_{t^n_i}\right]\\
        &\leq 5\,\mathbb{E}\left[\left(\widehat{Z}^n_{i,i+1}\right)^2\, \widehat{U}^n_{i,i+1}\,|\,\mathcal{F}^n_{t^n_i}\right]\,,\quad i \leq n-1\,,
    \end{align*}
    which concludes the proof of (\textit{iii}) using the tower property of conditional expectation.
\end{proof}

\subsection{Estimates on the sequences of processes}
We now aim to characterize the processes $\left(U^n\,, Z^n\right)_{n \geq 1}$ defined in \eqref{eq:def_process_U}-\eqref{eq:def_process_Z}. Firstly, we prove the martingality of $\left(Z^n\right)_{t \leq T}$ for any $n \geq 1\,$, along with a control in $U^n$ over its quadratic variation, which is obtained using the estimates given by Lemma \ref{lemma:properties_sequence}. Additionally, we show that $\left(Z^n\right)^2 - U^n$ is a also martingale, which will play an important role in identifying the limit.
\begin{lemma}
    \label{lemma:martingality}
    For any $n \geq 1\,$, we have:
    \begin{enumerate}[label = (\roman*)]
        \item $\left(Z^n_t\right)_{t \leq T}$ is a square-integrable martingale, such that 
        $$
        \left[Z^n\right]_t = \sum_{i = 0}^{\lfloor nt/T\rfloor - 1}\,\left(\widehat{Z}^n_{i,i+1}\right)^2\,, \quad t \leq T\,,
        $$
        as well as
        \begin{equation}
        \label{eq:exp_quad_Z_equals_exp_U}
        \mathbb{E}\left[\left[Z^n\right]_t - \left[Z^n\right]_s \,|\, \mathcal{F}^n_s \right] = \mathbb{E}\left[U^n_t - U^n_s \,|\, \mathcal{F}^n_s\right] \,, \quad s \leq t \leq T\,,
        \end{equation}
        and 
        \begin{equation}
        \label{eq:control_quad_Z}
        \mathbb{E}\left[\left(\left[Z^n\right]_t - \left[Z^n\right]_s\right)^2\right] \leq 25\, \mathbb{E}\left[\left(U^n_t - U^n_s\right)^2\right]\,, \quad s \leq t \leq T\,.
        \end{equation}
        \item $\left(\left(Z^n_t\right)^2 - U^n_t\right)_{t \leq T}$ is a martingale.
    \end{enumerate}
\end{lemma}
\begin{proof}
    In this proof, we consider $s \leq t \leq T\,$, and define $i := \lfloor nt/T \rfloor$ and $j := \lfloor ns/T \rfloor$ for readability. We compute
    \begin{align*}
        \mathbb{E}\left[Z^n_t \,|\, \mathcal{F}^n_s\right] &= \mathbb{E}\left[Z^n_{t^n_i}\,|\,\mathcal{F}^n_{t^n_j}\right] \\
        &= Z^n_{t^n_j} + \sum_{l = j}^{i-1}\, \mathbb{E}\left[\widehat Z^n_{l,l+1}\,|\, \mathcal{F}^n_{t^n_j}\right] \\
        &= Z^n_{t^n_j} + \sum_{l = j}^{i-1}\, \mathbb{E}\left[\mathbb{E}\left[\widehat Z^n_{l,l+1}\,|\, \mathcal{F}^n_{t^n_l}\right]\,|\, \mathcal{F}^n_{t^n_j}\right]
    \end{align*}
    using the tower property of conditional expectation, since $\mathcal{F}^n_{t^n_j} \subset \mathcal{F}^n_{t^n_l}$ for $j \leq l\,$. Using (\textit{i}) of Lemma \ref{lemma:properties_sequence}, we obtain
    $$
    \mathbb{E}\left[Z^n_t\,|\, \mathcal{F}^n_s\right] = Z^n_{t^n_j} = Z^n_s\,,
    $$
    proving that $\left(Z^n_t\right)_{t \leq T}$ is a martingale. The fact that it is square-integrable is direct from Lemma \ref{lemma:properties_sequence}.
    Furthermore, notice that $Z^n$ can be rewritten as 
    $$
    Z^n_t = \sum_{i = 0}^{\lfloor nt/T\rfloor - 1}\, \left(\widehat Z^n_{i,i+1}\right)^+ - \sum_{i = 0}^{\lfloor nt/T\rfloor - 1}\, \left(\widehat Z^n_{i,i+1}\right)^-\,, \quad t \leq T\,,
    $$
    which is the difference of two non-decreasing processes. Therefore it has finite variation on $[0\,,T]$ and is a quadratic pure-jump martingale \citep[Theorem~II.26]{protter2005stochastic}. This proves that
    $$
    \left[Z^n\right]_{t} = \sum_{0 \leq s \leq t}\, \left(\Delta Z^n_s\right)^2 = \sum_{i = 0}^{\lfloor nt/T \rfloor -1}\, \left(\widehat Z^n_{i,i+1}\right)^2\,, \quad t \leq T\,,
    $$
    since $Z^n_0 = 0\,$.
    From (\textit{ii}) of Lemma \ref{lemma:properties_sequence},
    \begin{align*}
    \mathbb{E}\left[\left(\widehat{Z}^n_{l,l+1}\right)^2\,|\, \mathcal{F}^n_{t_j^n}\right] &= \mathbb{E}\left[\mathbb{E}\left[\left(\widehat{Z}^n_{l,l+1}\right)^2\,|\,\mathcal{F}^n_{t_l^n}\right]\,|\, \mathcal{F}^n_{t_j^n}\right] \\
    &= \mathbb{E}\left[\mathbb{E}\left[\widehat{U}^n_{l,l+1}\,|\,\mathcal{F}^n_{t_l^n}\right]\,|\, \mathcal{F}^n_{t_j^n}\right] \\
    &= \mathbb{E}\left[\widehat{U}^n_{l,l+1}\,|\, \mathcal{F}^n_{t^n_j}\right]\,, \quad j \leq l \leq n-1\,,
    \end{align*}
    and therefore taking any $s \leq t \leq T\,$, we obtain
    \begin{align*}
        \mathbb{E}\left[\left[Z^n\right]_t - \left[Z^n\right]_s\, |\, \mathcal{F}^n_s\right] &= \sum_{l = j}^{i-1} \, \mathbb{E}\left[ \left(\widehat{Z}^n_{l,l+1}\right)^2\,|\, \mathcal{F}^n_{t_j^n}\right] \\
        &= \sum_{l = j}^{i-1} \mathbb{E}\left[\widehat{U}^n_{l,l+1}\,|\, \mathcal{F}^n_{t_j^n}\right]\\
        &= \mathbb{E}\left[U^n_t - U^n_s \,|\, \mathcal{F}^n_s\right]\,,
    \end{align*}
    where $i = \lfloor nt/T \rfloor$ and $j = \lfloor ns/T \rfloor\,$,
    which proves \eqref{eq:exp_quad_Z_equals_exp_U}. Finally, we have 
    \begin{align*}
        \mathbb{E}\left[\left(\left[Z^n\right]_t - \left[Z^n\right]_s\right)^2\right] &= \mathbb{E}\left[\left(\sum_{l = j}^{i-1}\,\left(\widehat{Z}^n_{l,l+1}\right)^2 \right)^2 \right] \\
        &= \sum_{l = j}^{i-1}\, \mathbb{E}\left[\left(\widehat{Z}^n_{l,l+1}\right)^4\right] + 2\,\sum_{l = j}^{i-1}\,\sum_{k = l+1}^{i-1}\, \mathbb{E}\left[\left(\widehat{Z}^n_{l,l+1}\right)^2\, \left(\widehat{Z}^n_{k,k+1}\right)^2\right]\,.
    \end{align*}
    For $l < k \leq n-1\,$, (\textit{ii}) of Lemma \ref{lemma:properties_sequence} gives 
    $$
\mathbb{E}\left[\left(\widehat{Z}^n_{l,l+1} \right)^2\, \left(\widehat{Z}^n_{k,k+1} \right)^2\right] = \mathbb{E}\left[\left(\widehat{Z}^n_{l,l+1}\right)^2\, \mathbb{E}\left[ \left(\widehat{Z}^n_{k,k+1}\right)^2\,|\,\mathcal{F}^n_{t_k^n}\right] \right] = \mathbb{E}\left[\left(\widehat{Z}^n_{l,l+1}\right)^2\, \widehat{U}^n_{k,k+1}\right] \,.
    $$
 Combining this with (\textit{iii}) of the same lemma, we obtain 
\begin{align*}
    \mathbb{E}\left[\left(\left[Z^n\right]_t - \left[Z^n\right]_s\right)^2\right] &\leq \, \mathbb{E}\left[5\,\sum_{l = j}^{i-1}\, \left(\widehat{Z}^n_{l,l+1}\right)^2\, \widehat{U}^n_{l,l+1} + 2\,\sum_{l = j}^{i-1}\,\sum_{k = l+1}^{i-1}\, \left(\widehat{Z}^n_{l,l+1}\right)^2\, \widehat{U}^n_{k,k+1}  \right] \\
    &\leq 5\, \mathbb{E}\left[\sum_{l = j}^{i-1}\, \left(\widehat{Z}^n_{l,l+1}\right)^2\, \left(\widehat{U}^n_{0,i} - \widehat{U}^n_{0,l}\right) \right] \\
    &\leq 5\, \mathbb{E}\,\left[\left(\left[Z^n\right]_t - \left[Z^n\right]_s\right)\,\left(U^n_t - U^n_s\right) \right]\,,
\end{align*}
using the fact that $U^n$ is non-decreasing. Applying Cauchy-Schwarz's inequality, we get 
$$
\mathbb{E}\left[\left(\left[Z^n\right]_t - \left[Z^n\right]_s\right)^2\right] \leq 5\, \sqrt{\mathbb{E}\left[\left(\left[Z^n\right]_t - \left[Z^n\right]_s\right)^2 \right]\, \mathbb{E}\left[\left(U^n_t - U^n_s\right)^2\right]}\,,
$$
which proves \eqref{eq:control_quad_Z}.

In order to prove (\textit{ii}) of Lemma \ref{lemma:martingality}, we notice that \eqref{eq:exp_quad_Z_equals_exp_U} shows that $\left(\left[Z^n\right]_t - U^n_t\right)_{t \leq T}$ is a martingale. Thus, $\left(\left(Z^n_t\right)^2 - U^n_t\right)_{t \leq T}$ is a local martingale. Moreover,
$$
\mathbb{E}\left[\sup_{0 \leq t \leq T}\, \left|Z^n_t\right|^2\right] \leq C\, \mathbb{E}\left[\left[Z^n\right]_T\right] = C \, \mathbb{E}\left[U^n_T\right] < + \infty\,,
$$
for some constant $C \geq 0\,$, using the BDG inequality, along with \eqref{eq:exp_quad_Z_equals_exp_U} with $t = T$ and $s = 0\,$, and the fact that the $\left(\widehat U^n_{i,i+1}\right)_{i \leq n-1}$ all have finite expectation in virtue of Lemma \ref{lemma:properties_sequence}. Adding the fact that 
$$
\mathbb{E}\left[\sup_{0 \leq t \leq T}\, U^n_t\right] = \mathbb{E}\left[U^n_T\right] < + \infty\,,
$$
since $U^n$ is non-decreasing. We conclude that $\mathbb E \left[\sup\limits_{0 \leq t \leq T}\, \left|\left(Z^n_t\right)^2 - U^n_t \right| \right] < + \infty\,$, and $\left(Z^n\right)^2 - U^n$ is therefore a true martingale.
\end{proof}
Before deriving estimates on our processes, we show that they satisfy a Volterra equation with a measure-valued kernel, which will be pivotal for the proofs.
\begin{lemma}
    \label{lemma:convolution_equation}
    For any $n \geq 1\,$, define the nonnegative finite measure $K^n := \sum_{i = 0}^{n-1}\,k^n_i\, \delta_{t^n_i}$ on $\left([0\,,T]\,, \mathcal B_{[0\,,T]}\right)\,$, where the $\left(k^n_i\right)_{i \leq n-1}$ are defined in \eqref{eq:kij}. Then we have 
    $$
    U^n_t = G_0^n(t) + \int_{[0,t]}\,\widetilde{Z}^n_{t-s}\, K^n(ds)\,, \quad t \leq T\,, \quad n \geq 1\,,
    $$
    where $\widetilde{Z}^n := b\,U^n + c\,Z^n$ as in \eqref{eq:tildeZ} and $G_0^n(t) := \int_0^{\lfloor \frac{nt}{T} \rfloor \frac{T}{n}}\,g_0(s)\,ds\,$.
\end{lemma}
\begin{proof}
    Let $n \geq 1\,$, $t < T$ and define $i := \lfloor nt/T \rfloor\,$, then $t^n_i \leq t < t^n_{i+1}\,$. Then,
    \begin{equation}
        \int_{[0,t]}\,\widetilde{Z}^n_{t-s}\,K^n(ds) = \sum_{j = 0}^{i}\,k^n_j\, \tilde{Z}^n_{t-t^n_j} \,.
    \end{equation}
    Since the partition $\left(t^n_j\right)_{j \leq n}$ is uniform,
    $$
    t^n_{i-j} = t^n_i - t^n_j \leq t - t^n_j < t^n_{i+1} - t^n_j = t^n_{i-j +1}\,, \quad j \leq i\,,
    $$
    and we can rewrite the integral as 
    \begin{align*}
        \int_{[0,t]}\, \widetilde{Z}^n_s\,K^n(ds) = \sum_{j = 0}^i\,k^n_j\, \widetilde{Z}^n_{t^n_{i-j}} = \sum_{j = 1}^i\,k^n_{i-j}\, \widetilde{Z}^n_{t^n_j} \,,
    \end{align*}
    using $\widetilde{Z}^n_0 = 0\,$. In the case $t = T\,$, we have
    $$
    \int_{[0,T]}\, \widetilde{Z}^n_{T-s}\, K^n(ds) = \sum_{j = 0}^{n-1}\, k^n_j\, \widetilde{Z}^n_{T - t^n_j} = \sum_{j = 0}^{n-1}\, k^n_j\, \widetilde{Z}^n_{t^n_{n-j}} = \sum_{j = 1}^n\, k^n_{n-j}\, \widetilde{Z}^n_{t^n_j}\,.
    $$
    Therefore, we always have
    \begin{equation}
    \label{eq:convolution_equality_general}
    \int_{[0,t]}\, \widetilde{Z}^n_{t-s}\, K^n(ds) = \sum_{j = 1}^{i }\, k^n_{i-j}\,\widetilde{Z}^n_{t^n_j}\,, \quad t \leq T\,,
    \end{equation}
    with $i = \lfloor nt/T \rfloor \,$.

    Taking $t \leq T\,$, setting $i := \lfloor nt/T \rfloor\,$, and using \eqref{eq:alphai}-\eqref{eq:Zii}, we can compute
    \begin{align*}
        U^n_t &= \sum_{j = 0}^{i-1}\, \widehat U^n_{i,i+1} \\
        &= \sum_{j = 0}^{i-1}\, \left(\alpha^n_j + k_0^n\, \left(b\,\widehat U^n_{j,j+1} + c\, \widehat Z^n_{j,j+1}\right) \right) \\
        &= \int_0^{t^n_i}\, g_0(s)\,ds + \sum_{j  = 0}^{i-1}\, \sum_{l = 0}^{j}\, k^n_{j-l}\, \left(b\, \widehat{U}^n_{l,l+1} + c \,\widehat{Z}^n_{l,l+1}\right)\,.
    \end{align*}
    Using the fact that $\int_0^{t^n_i}\, g_0(s)\,ds = G_0^n(t)$ and $b\,\widehat U^n_{l,l+1} + c \,\widehat{Z}^n_{l,l+1} = \widetilde{Z}^n_{t^n_{l+1}} - \widetilde{Z}^n_{t^n_l}\,$, we get
    \begin{align*}
        U^n_t &= G_0^n(t) + \sum_{j = 0}^{i-1}\, \sum_{l = 0}^j\,k^n_{j -l}\, \left(\widetilde{Z}^n_{t^n_{l+1}} - \widetilde{Z}^n_{t^n_l}\right) \\
        &= G_0^n(t) + \sum_{l = 1}^i \, \widetilde{Z}^n_{t^n_l}\, \sum_{j = l}^i\, k^n_{j-l} - \sum_{l = 1}^{i-1}\, \widetilde{Z}^n_{t^n_l}\, \sum_{j = l}^{i-1}\,k^n_{j-l}\\
        &= G_0^n(t) + \sum_{l = 1}^i \, \widetilde{Z}^n_{t^n_l}\, k^n_{i-l}\,.
    \end{align*}
    Therefore, applying \eqref{eq:convolution_equality_general}, we finally obtain
    $$
    U^n_t = G_0^n(t) + \int_{[0,t]}\, \widetilde{Z}^n_{t-s}\,K^n(ds)\,.
    $$
\end{proof}
We are now ready to state our estimates on the processes $\left(U^n\,, Z^n\right)_{n \geq 1}\,$. We introduce the notation
\begin{equation}
    \label{eq:def_G_0}
    G_0(t) := \int_0^t\, g_0(s)\,ds\,, \quad t \leq T\,.
\end{equation}
\begin{lemma}
    We have the upperbound 
    \begin{equation}
        \label{eq:bound_exp_U}
        \mathbb{E}\left[U^n_t\right] \leq G_0(T)\,, \quad t \leq T\,, \quad n \geq 1\,,
    \end{equation}
    and there exists constants $C_1\,, C_2 \geq 0\,$, such that
    \begin{equation}
        \label{eq:bound_exp_U_2}
        \mathbb{E}\left[\left(U^n_t\right)^2\right] \leq C_1\,G_0(T) + C_2\, G_0(T)^2\,, \quad t \leq T\,, \quad n \geq 1\,.
    \end{equation}
    Moreover, 
    \begin{equation}
        \label{eq:bound_variation_U}
        U^n_{t^n_i} - U^n_{t^n_j} \leq G_0(t^n_i) - G_0(t^n_j) + 2 \, \left\|\widetilde{Z}^n\right\|_{\infty}\,\int_0^{t^n_i - t^n_j}\, K(s)\,ds\,, \quad j \leq i \leq n\,, \quad n \geq 1\,.
    \end{equation}
\end{lemma}
\begin{remark}
    Note that 
    \begin{equation}
    \label{eq:bound_exp_Z_2}
    \mathbb{E}\left[\sup_{0 \leq t \leq T}\,\left|Z^n_t\right|^2\right] \leq C_3\, \mathbb{E}\left[\left[Z^n\right]_T\right] = C_3 \, \mathbb{E}\left[U^n_T\right] \leq C_3\,G_0(T) \,, \quad n \geq 1\,,
    \end{equation}
    for a certain $C_3 \geq 0\,$, using the BDG inequality, along with \eqref{eq:exp_quad_Z_equals_exp_U} with $t = T$ and $s = 0\,$, and \eqref{eq:bound_exp_U}. Additionally, 
    \begin{equation}
    \label{eq:bound_exp_Z_tilde}
\mathbb{E}\left[\left\|\tilde{Z}^n\right\|_{\infty}\right] \leq |b|\, \mathbb{E}\left[U^n_T\right] + c \, \left(1 + \mathbb{E}\left[\sup_{0 \leq t \leq T}\, \left|Z^n_t\right|^2\right]\right) \leq C_4\, \left(1 + G_0(T)\right)\,, \quad n \geq 1\,,
    \end{equation}
    for a certain constant $C_4 \geq 0\,$, using \eqref{eq:bound_exp_Z_2} and \eqref{eq:bound_exp_U}.
\end{remark}
\begin{proof}
    We make use of the Volterra equation given by Lemma \ref{lemma:convolution_equation}. Since $b \leq 0$ and $U^n$ is a nonnegative process for any $n \geq 1\,$, we have 
    $$
    U^n_t \leq G_0^n(t) + c\,\int_{[0,t]}\,Z^n_{t-s}\,K^n(ds)\,, \quad t \leq T\,, \quad n \geq 1\,.
    $$
    Since $Z^n$ is a square-integrable martingale for any $n \geq 1\,$, taking the expectation in the previous relation proves \eqref{eq:bound_exp_U}, using the fact that
    $$
    G_0^n(t) \leq G_0^n(T) = G_0(T) \,, \quad t \leq T \,, \quad n \geq 1\,.
    $$
    Moreover, we have 
    $$
    U^n_t \leq G_0(T) + c\,\sup_{0 \leq s \leq T}\, \left|Z^n_s\right|\, \int_{[0,T]}\, K^n(ds)\, \quad t \leq T\,, \quad n \geq 1\,,
    $$
    so that
    $$   \mathbb{E}\left[\left(U^n_t\right)^2\right] \leq 2\,G_0(T)^2 + 2\,\left(c\,\int_0^T\,K(s)\,ds\right)^2 \, \mathbb{E}\left[\sup_{0 \leq s \leq T}\, \left|Z^n_s\right|^2\right]\,, \quad t \leq T\,, \quad n \geq 1\,,
    $$
    using the fact that $\int_{[0,T]}\,K^n(ds) = \int_0^T\,K(s)\,ds$ for any $n \geq 1\,$. Applying \eqref{eq:bound_exp_Z_2} (which only uses \eqref{eq:bound_exp_U}) leads to 
    $$
    \mathbb{E}\left[\left(U^n_t\right)^2\right] \leq 2\,G_0(T)^2 + 2\,C_3\,G_0(T)\, \left(c\,\int_0^T\, K(s)\,ds\right)^2\,, \quad t \leq T\,, \quad n \geq 1\,,
    $$
    and proves \eqref{eq:bound_exp_U_2}.

    In order to prove \eqref{eq:bound_variation_U}, we fix $n \geq 1$ and $j \leq i \leq n\,$. Lemma \ref{lemma:convolution_equation} gives
    \begin{align*}
        U^n_{t^n_i} - U^n_{t^n_j} &= G_0^n(t^n_i) - G_0^n(t^n_j) + \sum_{l = 1}^{i}\, k^n_{i-l}\, \widetilde{Z}^n_{t^n_l} - \sum_{l = 1}^j \,k^n_{j-l}\, \widetilde{Z}^n_{t^n_l} \\
        &= G_0^n(t^n_i) - G_0^n(t^n_j) + \sum_{l = 1}^j\, \widetilde{Z}^n_{t^n_l}\, \left(k^n_{i-l} - k^n_{j-l} \right) + \sum_{l = j+1}^i\, k^n_{i-l}\, \widetilde{Z}^n_{t^n_l}\,.
    \end{align*}
    In virtue of Assumption \ref{assumption:preserve_nonnegativity}, $\bar K^n$ is non-decreasing and therefore $\left(k^n_i\right)_{i \leq n-1}$ is a non-decreasing sequence. We deduce that 
    \begin{align*}
        U^n_{t^n_i} - U^n_{t^n_j} \leq G_0^n(t^n_i) - G_0^n(t^n_j) + \left\| \widetilde{Z}^n \right\|_{\infty}\, \left(\sum_{l = 1}^j\, k^n_{j-l} - k^n_{i-l} + \sum_{l = j+1}^{i}\,k^n_{i-l} \right) \,.
    \end{align*}
    We conclude by noticing that $G_0^n(t^n_l) = G_0(t^n_l)$ for any $l \leq n\,,$ and by upperbounding
    \begin{align*}
        \sum_{l = 1}^j\, k^n_{j-l} - k^n_{i-l} + \sum_{l = j+1}^{i}\,k^n_{i-l} &= \int_0^{t^n_j}\, K(s)\,ds - \int_{t^n_i - t^n_j}^{t^n_i}\,K(s)\,ds + \int_0^{t^n_i - t^n_j}\, K(s)\,ds \\
        &\leq \int_0^{t^n_i}\, K(s)\,ds - \int_{t^n_i - t^n_j}^{t^n_i}\,K(s)\,ds + \int_0^{t^n_i - t^n_j}\, K(s)\,ds \\
        &\leq 2\, \int_{0}^{t^n_i - t^n_j}\, K(s)\,ds\,.
    \end{align*}
\end{proof}

\subsection{Tightness}
We want to prove C-tightness of our sequences of processes, i.e.~tightness for the Skorokhod $J_1$ topology, with continuous limits. We use the following criterion.
\begin{proposition}[\cite{jacod2013limit}, Proposition~VI.3.26]
    \label{proposition:C_tight_criterion}
    A sequence of càdlàg stochastic processes $\left(X^n\right)_n$ defined on $[0\,,T]$ is C-tight if the following two conditions are satisfied:
    \begin{itemize}
        \item[$\bullet$] $\lim\limits_{R \to + \infty}\, \limsup\limits_n\, \mathbb{P}\left(\left\|X^n\right\|_{\infty} > R \right) = 0\,$.
        \item[$\bullet$] For any $\varepsilon > 0\,$,
        $$\lim\limits_{\delta \to 0^+}\, \limsup\limits_n \, \mathbb{P}\left(w\left(X^n\,,\delta\right) \geq \varepsilon\right) = 0\,,$$
        where $w$ is the modulus of continuity
        $$
        w(f, \delta) := \sup_{\substack{0 \leq s \leq t \leq T \\ |t - s| \leq \delta}}\, |f(t) - f(s)|\,, \quad f \in D([0\,,T])\,, \quad \delta > 0\,.
        $$
    \end{itemize}
\end{proposition}
We prove C-tightness of both sequences of processes.
\begin{lemma}
    \label{lemma:C_tightness}
    Both $\left(U^n\right)_{n \geq 1}$ and $\left(Z^n\right)_{n \geq 1}$ are C-tight.
\end{lemma}
\begin{proof}
    We first prove tightness of $\left(U^n\right)_{n \geq 1}\,$. \eqref{eq:bound_exp_U} gives us 
    $$
    \sup_n \, \mathbb{E}\left[\left\|U^n\right\|_{\infty}\right] = \sup_n \, \mathbb{E}\left[U^n_T\right] \leq G_0(T) < + \infty\,,
    $$
    so that the first condition of Proposition \ref{proposition:C_tight_criterion} is satisfied by Markov's inequality.
    
    Then, by definition, 
    $$
    w\left(U^n\,,\delta\right) = \max_{\substack{0 \leq j \leq i \leq n \\ |t^n_i - t^n_j| \leq \delta}} \, U^n_{t^n_i} - U^n_{t^n_j} \,, \quad n \geq 1\,,\quad \delta > 0\,.
    $$
    Thus, for any $\delta > 0$ and $n \geq 1\,$, \eqref{eq:bound_variation_U} gives 
    $$
    w\left(U^n\,, \delta\right) \leq w\left(G_0\,,\delta\right) + 2\,\left\|\tilde Z^n \right\|_{\infty}\, \int_0^{\delta}\,K(s)\,ds\,,
    $$
    so that
    $$
\mathbb{E}\left[w\left(U^n\,,\delta\right)\right] \leq w(G_0\,, \delta) + 2\,C_4\, \left(1 + G_0(T)\right)\, \int_0^{\delta}\,K(s)\,ds\,,
    $$
    using \eqref{eq:bound_exp_Z_tilde}. Since $g_0 \in L^1_{\rm loc}\left(\mathbb{R}_+, \mathbb R_+\right)\,$, $G_0$ defined in \eqref{eq:def_G_0} is continuous on $[0\,,T]\,$, by the dominated convergence theorem. With the same argument, $\int_0^{\delta}\,K(s)\,ds$ vanishes as $\delta$ goes to $0$. This proves C-tightness of $\left(U^n\right)_{n \geq 1}$ by Markov's inequality.

    For the sequence of martingales $\left(Z^n\right)_{n \geq 1}\,$, we first have
    $$
    \mathbb{E}\left[\sup_{0 \leq t \leq T}\, \left|Z^n_t\right|\right] \leq 1 + C_3\,G_0(T) < + \infty\,,
    $$
    from \eqref{eq:bound_exp_Z_2}, thus the first condition for C-tightness is satisfied. For the second criterion, let us fix $\delta > 0$ and $n > T / \delta\,$, we define $m := \lceil n\delta / T\rceil$ and $v := \lceil n / m \rceil\,$, as well as 
    $$
    s_i := i \, \frac{mT}{n} = t^n_{im}\,, \,\, 0 \leq i \leq v-1\,, \quad \text{and} \quad s_v := T \,.
    $$
    One can check that the two following inequalities are satisfied:
    \begin{equation}
        \label{eq:partition_above_delta}
        s_i - s_{i-1} \geq \delta \,, \quad 1 \leq i \leq v-1\,,
    \end{equation}
    \begin{equation}
        \label{eq:partition_upper_bound}
        s_i - s_{i-1} \leq \delta + T / n \leq 2\,\delta\,, \quad 1 \leq i \leq v\,.
    \end{equation}
Under the condition \eqref{eq:partition_above_delta}, \citet[Theorem~7.4]{billingsley2013convergence} gives 
$$
w\left(Z^n\,, \delta\right) \leq 3\, \max_{1 \leq i \leq v}\, \sup_{s_{i-1} \leq s \leq s_i}\, \left|Z^n_s - Z^n_{s_{i-1}}\right|\,.
$$
Therefore, for any $\varepsilon > 0\,$, recalling that our random variables admit finite moments of order $4$ thanks to Lemma \ref{lemma:properties_sequence},
\begin{align}
\label{eq:bdg_Z4}
    \mathbb{P}\left( w\left(Z^n\,, \delta\right) \geq 3\,\varepsilon\right) &\leq \sum_{i = 1}^v \, \mathbb{P}\left( \sup_{s_{i-1} \leq s \leq s_i}\, \left|Z^n_s - Z^n_{s_{i-1}}\right| \geq \varepsilon\right) \nonumber\\
    &\leq \frac{1}{\varepsilon^4}\, \sum_{i = 1}^v \, \mathbb{E}\left[\sup_{s_{i-1} \leq s \leq s_i}\, \left|Z^n_s - Z^n_{s_{i-1}}\right|^4 \right] \nonumber\\
    &\leq \frac{C_5}{\varepsilon^4}\, \sum_{i = 1}^v \, \mathbb{E}\left[\left(\left[Z^n\right]_s - \left[Z^n\right]_{s_{i-1}}\right)^2\right]\,,
\end{align}
for a certain constant $C_5 \geq 0\,$, thanks to the BDG inequality. We can apply \eqref{eq:control_quad_Z} to obtain
\begin{align*}
    \mathbb{P}\left( w\left(Z^n\,, \delta\right) \geq 3\,\varepsilon\right) &\leq \frac{25\,C_5}{\varepsilon^4}\, \sum_{i = 1}^v \, \mathbb{E}\left[\left(U^n_{s_i} - U^n_{s_{i-1}}\right)^2\right] \\
    &\leq \frac{25\, C_5}{\varepsilon^4}\, \mathbb{E}\left[U^n_T\, \max_{1 \leq i \leq v}\, U^n_{s_i} - U^n_{s_{i-1}}\right]\,.
\end{align*}
Finally, using \eqref{eq:bound_variation_U}, since $\left(s_i^n\right)_{i \leq v}$ is a sub-partition of $\left(t^n_i\right)_{i \leq n}$ with step less than $2\delta$ (see \eqref{eq:partition_upper_bound}), we obtain
\begin{align*}
    \mathbb{P}\left(w\left(Z^n\,, \delta\right) \geq 3\, \varepsilon \right) \leq \frac{25\, C_5}{\varepsilon^4}\, \left(w\left(G_0\,, 2\delta\right) + 2 \, \int_0^{2\delta}\, K(s)\,ds\, \mathbb{E}\left[U^n_T\, \left\| \tilde Z^n \right\|_{\infty}\right] \right)\,.
\end{align*}
We conclude by noticing that 
$$
\sup_n \, \mathbb{E}\left[U^n_T\, \left\| \tilde Z^n \right\|_{\infty}\right] \leq \sup_n \, \left(|b| \,\mathbb{E}\left[\left(U^n_T\right)^2\right] + c\, \sqrt{\mathbb{E}\left[\left(U^n_T\right)^2\right]\, \mathbb{E}\left[\sup_{0 \leq t \leq T}\, \left|Z^n_t\right|^2\right]}\right) < + \infty\,,
$$
combining \eqref{eq:bound_exp_U_2} and \eqref{eq:bound_exp_Z_2}. Thus, $\left(Z^n\right)_{n \geq 1}$ is C-tight.
\end{proof}
\begin{remark}
\label{remark:joint_tightness}
    In particular, the two sequences are $J_1(\mathbb{R})$-tight, which implies joint $J_1\left(\mathbb{R}^2\right)$-tightness of $\left(U^n\,, Z^n\right)_{n \geq 1}\,$, see \citet[Lemma~3.2]{whitt2007proofs}.
\end{remark}

\subsection{Characterization of the limiting processes}
Thanks to Prokhorov's theorem, joint tightness given by Lemma \ref{lemma:C_tightness} and Remark \ref{remark:joint_tightness} gives the existence of a weakly convergent subsequence of $\left(U^n\,,Z^n\right)_{n \geq 1}$ in $\left(D([0\,,T])^2\,, J_1\left(\mathbb{R}^2\right)\right)\,$. We therefore aim to characterize the law of those limits. This is achieved using the following stability lemma for convolutions.
\begin{lemma}
    \label{leamm:stability_convolution}
Consider a sequence $\left(f_n\right)_{n \geq 1}$ of càdlàg functions on $[0\,,T]\,$, converging in $L^1([0\,,T])$ to $f$ continuous. If the sequence $\left(F_n\right)_{n \geq 1}$ defined by 
$$
F_n(t) := \int_{[0,t]} \, f_n(t-s)\, K^n(ds)\,, \quad t \leq T\,, \quad n \geq 1\,,
$$
where $K^n$ is defined in Lemma \ref{lemma:convolution_equation}, converges in $L^1([0\,,T])$ to $F$ continuous, then 
$$
F(t) = \int_0^t \, f(t-s)\, K(s)\,ds \,, \quad t \leq T\,.
$$
\end{lemma}
\begin{proof}
    For any $n \geq 1\,$, we can write 
    \begin{align*}
        \int_0^T\, \left|F(t) - \int_0^t \, f(t-s)\, K(s)\, ds\right| \, dt &\leq \int_0^T \, \left|F(t) - F_n(t)\right| \, dt \\
        &\quad + \int_0^T\, \left| \int_{[0,t]} \, \left(f_n(t-s) - f(t-s)\right)\, K^n(ds) \right|\, dt \\
        &\quad + \int_0^T \, \left| \int_{[0,t]} \, f(t-s) \, K^n(ds) - \int_0^t \, f(t-s) \, K(s) \, ds\right| \, dt\,.
    \end{align*}
The first term on the right-hand side vanishes as $n$ goes to infinity by definition. The second is upper-bounded by 
$$
\int_0^T\, \int_s^T \, \left|f_n(t-s) - f(t-s)\right|\, dt \, K^n(ds) \leq \int_0^T \, K(s)\,ds \, \int_0^T \, \left|f_n(t) - f(t)\right|\, dt\,,
$$
which also vanishes by convergence in $L^1([0\,,T])\,$. Finally, we introduce the notations
$$
I^n_t := \int_{[0,t]}\,f(t-s)\,K^n(ds)\,, \quad I_t := \int_0^t\,f(t-s)\,K(s)\,ds\,, \quad t \leq T\,, \quad n \geq 1\,.
$$
For any $n \geq 1\,$, $0 < t < T\,$, setting $i := \lfloor nt/T \rfloor\,$, we have
$$
I^n_t = \sum_{j = 0}^{i} \,f(t- t^n_j)\,k^n_j = \sum_{j = 0}^{i}\,\int_{t^n_j}^{t^n_{j+1}}\,f(t - t^n_j)\, K(s)\,ds\,,
$$
and 
$$
I_t = \sum_{j = 0}^{i}\, \int_{t^n_j}^{t^n_{j+1}}\, f(t-s)\,K(s)\,ds - \int_t^{t^n_{i+1}}\,f(t-s)\,K(s)\,ds\,.
$$
Therefore, 
\begin{align*}
    \left|I^n_t - I_t\right| &\leq \sum_{j = 0}^i\, \int_{t^n_j}^{t^n_{j+1}}\, \left|f(t-t^n_j) - f(t-s)\right|\,K(s)\,ds + \int_{t}^{t^n_{i+1}}\, |f(t-s)|\,K(s)\,ds \\
    &\leq w\left(f\,,\frac{T}{n}\right)\, \sum_{j = 0}^i\, \int_{t^n_j}^{t^n_{j+1}}\, K(s)\,ds + \left\|f\right\|_{\infty}\, \int_t^{t^n_{i+1}}\,K(s)\,ds \\
    &\leq w\left(f\,,\frac{T}{n}\right)\, \int_0^T\,K(s)\,ds + \left\|f\right\|_{\infty}\, \int_t^{t^n_{i+1}}\,K(s)\,ds\,.
\end{align*}
Since $f$ is continuous, the modulus vanishes in the limit $n \to \infty\,$. Furthermore, $t^n_{i+1} = \frac{T}{n}\left(1 + \lfloor \frac{nt}{T}\rfloor \right)$ converges to $t\,$. By the dominated convergence theorem $\int_t^{t^n_{\lfloor nt/T\rfloor +1}}\,K(s)\,ds$ goes to $0\,$, as a consequence of $K$ being in $L^1([0\,,T])\,$. We deduce that for any $0 < t <T\,$, $I_t^n$ converges to $I_t$ as $n$ goes to infinity. Using the domination
$$
\left| I_t^n - I_t \right| \leq 2\, \|f\|_{\infty}\, \int_0^T \, K(s)\,ds < + \infty\,, \quad t \leq T\,, \quad n \geq 1\,,
$$
the dominated convergence theorem shows that the third integral on the right-hand side of our first inequality also vanishes in the limit. We conclude that 
$$
\int_0^T\, \left|F(t) - \int_0^t\, f(t-s)\,K(s)\,ds\right|\,dt = 0\,,
$$
and $F(t) = \int_0^{t}\,f(t - s)\,K(s)\,ds$ for $t \leq T$ in virtue of the continuity assumption on $F$ and $f\,$.
\end{proof}
We can now characterize the accumulation points of our sequence of processes.
\begin{lemma}
    \label{lemma:characterization_limit}
    Let $\left(U\,, Z\right)$ be a weak limit of a convergent subsequence $\left(U^{n_k}\,, Z^{n_k}\right)_{k \geq 1}\,$. Then:
    \begin{itemize}
        \item[$\bullet$] $U$ is a continuous, non-decreasing, nonnegative process, starting from 0.
        \item[$\bullet$] $Z$ is a continuous square-integrable martingale with respect to the filtration generated by $\left(U\,,Z\right)\,$, starting from 0, such that 
        $$
        \langle Z \rangle_t = U_t\,, \quad t \leq T. 
        $$
        \item[$\bullet$] The following Volterra equation is satisfied:
        $$
        U_t = \int_0^t\,g_0(s)\,ds + \int_0^t\, K(t-s)\, \left(b\,U_s + c\,Z_s\right)\,ds\,, \quad t \leq T\,, \quad \text{a.s.}
        $$
    \end{itemize}
\end{lemma}
\begin{proof}
    By Skorokhod's representation theorem \citep[Theorem~1.6.7]{billingsley2013convergence}, there exists càdlàg processes $\left(\mathfrak{U}^k\,, \mathfrak Z^k\right)_{k \geq 1}$ and $\left(\mathfrak U\,, \mathfrak Z\right)$ all defined on a common probability space $\left(\Omega'\,, \mathcal F'\,, \mathbb P'\right)\,$, such that:
    \begin{itemize}
        \item[$\bullet$] For any $k\geq 1\,$, $\left(U^{n_{k}}\,, Z^{n_{k}}\right) \sim \left( \mathfrak U^k \,, \mathfrak Z^k\right)\,$, and $\left(U\,, Z\right) \sim \left(\mathfrak U\,, \mathfrak Z\right)\,$.
        \item[$\bullet$] For any $\omega \in \Omega'\,$, $\left(\mathfrak U^k\,, \mathfrak Z^k\right)(\omega) \overset{J_1\left(\mathbb{R}^2\right)}{\longrightarrow} \left(\mathfrak U\,, \mathfrak Z\right)(\omega)\,$.
    \end{itemize}
    Since $\left(U^n\right)_{n \geq 1}$ and $\left(Z^n\right)_{n \geq 1}$ are both C-tight, $U$ and $Z$ are almost surely continuous, and therefore this is also the case for $\mathfrak U$ and $\mathfrak Z\,$.~Additionally, for any $n \geq 1\,$, $U^n$ is almost surely non-decreasing, nonnegative, and starts from $0\,$, so that these properties are also satisfied by $\mathfrak U^k$ for $k \geq 1\,$. Furthermore, $J_1$-convergence to a continuous function is equivalent to local uniform convergence, see \citet[Proposition~VI.1.17.b]{jacod2013limit}.~We deduce that $\mathfrak U$ is almost surely non-decreasing, nonnegative and starting from $0\,$. This is also the case for $U\,$, since $U \sim \mathfrak U\,$. The same reasoning shows that $Z$ almost surely starts from $0\,$. Additionally, we have,
    $$
    \mathfrak U^k_t(\omega) = G_0^k(t) + \int_{[0,t]}\, \left(b\, \mathfrak U^k_{t-s}(\omega) + c\, \mathfrak Z^k_{t-s}(\omega)\right)\, K^n(ds)\,, \quad t \leq T\,, \quad k \geq 1\,,
    $$
thanks to Lemma \ref{lemma:convolution_equation}, along with the càdlàg property for $\left(\mathfrak U^k(\omega)\,, \mathfrak Z^k(\omega)\right)_{k \geq 1}\,$, for any $\omega$ in $\mathcal A \in \mathcal F'$ of unit measure.~Therefore, taking $\left(F_k\right)_{k \geq 1} := \left(\mathfrak U^k(\omega) - G_0^k\right)_{k \geq 1}$\footnote{By the dominated convergence theorem, $G_0^k(t)$ converges to $G_0(t)$ for any $t \leq T\,$, and is dominated by $G_0(T)\,$. Thus, $G_0^k$ converges to $G_0$ in $L^1([0\,,T])\,$.} and $\left(f_k\right)_{k \geq 1} := \left(b\, \mathfrak U^k(\omega) + c\, \mathfrak Z^k(\omega)\right)_{k \geq 1}\,$, in Lemma \ref{leamm:stability_convolution}, since the $J_1$ topology is stronger than the $L^1$ topology \citep[Theorem~11.5.1]{whitt2002stochastic}, we obtain 
$$
\mathfrak U_t(\omega) = \int_0^t\, g_0(s)\,ds + \int_0^t\, K(t-s)\, \left(b\, \mathfrak U_s(\omega) + c\, \mathfrak Z_s(\omega) \right)\, ds\,, \quad t \leq T\,,
$$
for any $ \omega \in \mathcal A\,\cap\,\left\{\omega\,:\, \text{$\mathfrak U(\omega), \mathfrak Z(\omega)$ continuous} \right\}.$
Since $\left(U\,,Z\right) \sim \left(\mathfrak U\,, \mathfrak Z\right)\,$,
$$
U_t = \int_0^t \, g_0(s)\,ds + \int_{0}^t\, K(t-s)\, \left(b\,U_s + c\, Z_s\right)\, ds \,, \quad t \leq T\,,
$$
almost surely.

It remains to characterize the limiting process $Z\,$. For any $n \geq 1\,$, $Z^n$ is a martingale with respect to $\left(\mathcal F^n_t\right)_{t \leq T}\,$ generated by $U^n$ (see \eqref{eq:def_continuous_filtration}). Since for $i \leq n-1\,$, Algorithm \ref{alg:simulation} defines $\widehat Z^n_{i,i+1}$ as a linear combination of the $\left(\widehat U^n_{j,j+1}\right)_{j \leq i}\,$, in our case the filtration generated by $U^n$ coincides with the one generated by $\left(U^n\,, Z^n\right)\,$. Thus, $Z^n$ is also a martingale with respect to the latter filtration. According to \citet[Theorem~5.3]{whitt2007proofs}\footnote{The martingale property is proved for continuity points of the limit. Since our limits are almost surely continuous, it shows full martingality in our case.}, if we prove that $\left(Z^n_t\right)_{n \geq 1}$ is uniformly integrable for any $t \leq T\,$, then the weak limit $Z$ is a martingale with respect to the filtration generated by $\left(U,Z\right)\,$. Using \eqref{eq:bound_exp_Z_2}, we have
    $$
    \sup_n \,\sup_{0 \leq t \leq T}\, \mathbb{E}\left[ \left|Z^n_t\right|^2\right] \leq C_3\,G_0(T) < + \infty\,,
    $$
    which shows that $\left(Z^n_t\right)_{ t \leq T, n \geq 1}$ is uniformly integrable. Therefore $Z$ is a martingale with respect to the filtration generated by $\left(U\,, Z\right)\,$. Furthermore, 
    $$
    \mathbb{E}\left[\sup_{0 \leq t \leq T}\left|Z^n_t\right|^4\right] \leq C_5 \, \mathbb{E}\left[\left[Z^n\right]_T^2\right] \leq 25\, C_5\, \mathbb{E}\left[\left(U^n_T\right)^2\right]\,, \quad n \geq 1\,,
    $$
    using respectively the BDG inequality on $\left|Z^n\right|^4$ (see \eqref{eq:bdg_Z4})  and \eqref{eq:control_quad_Z}. Thus,
    $$
    \sup_n\, \sup_{0 \leq t \leq T}\,\mathbb{E}\left[ \left|\left(Z^n_t\right)^2 - U^n_t\right|^2\right] \leq 2\, \left(1 + 25\, C_5\right)\, \sup_n \,\mathbb{E}\left[\left(U^n_T\right)^2\right] < + \infty\,,
    $$
    thanks to \eqref{eq:bound_exp_U_2}. This proves that $\left(\left(Z^n_t\right)^2 - U^n_t\right)_{t \leq T, n \geq 1}$ is uniformly integrable, and since $\left(Z^n\right)^2 - U^n$ is a martingale for each $n \geq 1\,$, with respect to the filtration generated by $\left(U^n\,,Z^n\right)$ using the same reasoning, we deduce that $Z^2 - U$ is also a martingale with respect to the filtration generated by $\left(U\,,Z\right)$. This first proves that $Z$ is square-integrable, and by continuity and therefore predictablity of $U$, we deduce that 
    $$
    \langle Z \rangle_t = U_t \,, \quad t \leq T\,.
    $$
\end{proof}

We can now conclude with the proof of Theorem \ref{theorem:weak_convergence}.
\begin{proof}[Proof of Theorem \ref{theorem:weak_convergence}]
The result is direct from the joint tightness given by Remark \ref{remark:joint_tightness} and the characterization of the accumulation points of Lemma \ref{lemma:characterization_limit}.
\end{proof}

\appendix
\section{Inverse Gaussian and its sampling}\label{A:IG}

The Inverse Gaussian distribution, also known as the Wald distribution, is a continuous probability distribution on $\R_+$ with two parameters:    $ \mu > 0$ (mean parameter) and $\lambda > 0$ (shape parameter). The probability density function  of the Inverse Gaussian distribution is given by:
\[
f(x ; \mu, \lambda) = \sqrt{\frac{\lambda}{2\pi x^3}} \exp \left( -\frac{\lambda (x - \mu)^2}{2\mu^2 x} \right), \quad x > 0.
\]
We  denote  $X\sim IG(\mu, \lambda)$ to indicate that $X$ is a random variable with an inverse Gaussian distribution with mean parameter $\mu$ and shape parameter $\lambda$. Its characteristic function is given by 
\begin{align}\label{eq:IGchar}
\mathbb E \left [ \exp\left( wX\right) \right]= \exp\left(\frac{\lambda}{\mu} \left(1 - \sqrt{1 - \frac{2w\mu^2 }{\lambda}}\right)\right), 
\end{align}
for all $w\in \mathbb C$ such that  $\Re(w)\leq 0$.  In addition, it admits moments of any order, satisfying
\begin{equation}
\label{eq:IGmean}
    \mathbb{E}\left[X\right] = \mu\,, \quad \text{and} \quad \mathbb{E}\left[X^n\right] = (2n - 3)\, \frac{\mu^2}{\lambda}\,\mathbb{E}\left[X^{n-1}\right] + \mu^2\, \mathbb{E}\left[X^{n-2}\right]\,, \quad n \geq 2\,.
\end{equation}

Inverse Gaussian random variables can be simulated easily using one Gaussian random variable and one Uniform random variable using  an acceptance-rejection step as shown in \citet*{michael1976generating}, we recall the algorithm here.

\begin{algorithm}[H]
\caption{Sampling from the Inverse Gaussian Distribution}\label{alg:IG_sampling}
\begin{algorithmic}[1]
\State \textbf{Input:} Parameters \( \mu > 0 \), \( \lambda > 0 \).
\State \textbf{Output:} Sample \( IG \) from the Inverse Gaussian distribution.

\State   
Generate \( \xi \sim \mathcal{N}(0, 1) \) and compute \( Y = \xi^2 \).

\State Compute the candidate value \( X \):  
$$
X = \mu + \frac{\mu^2 Y}{2\lambda} - \frac{\mu}{2\lambda} \sqrt{4\mu\lambda Y + \mu^2 Y^2}.
$$

\State Generate a uniform random variable:  
Sample \( \eta \sim \text{Uniform}(0, 1) \).

\State Select the output:  
\If{$ \eta \leq \frac{\mu}{\mu + X} $}
    \State Set the output \( IG = X \).
\Else
    \State Set \( IG = \frac{\mu^2}{X} \).
\EndIf
\end{algorithmic}
\end{algorithm}

\section{Volterra Heston's characteristic function}\label{S:Heston}

We recall the expression of the joint characteristic function of the `integrated variance' $U$ and $\log S$ in the Volterra Heston model in \eqref{eq:HestonS} with $U$ given by \eqref{eq:U0}.  Let $v,w \in  \mathbb{C}$ be such that
\begin{equation}\label{eq:hestonchar}
    \Re w \leq 0 \quad \text{and}\quad  0 \leq \Re v\leq 1.
\end{equation}
The joint conditional characteristic function of $(\log S, U)$ is given by 
\begin{equation} \label{eq:hestonchar}
 \mathbb{E} \left[ \left. \exp\left(  v \log \frac{S_T}{ S_t} + w U_{t,T}\right) \right| \mathcal F_t \right] = \exp\left(  \int_t^T\!  \! F(\psi(T-s))dG_t(s) \right), \quad t \leq T,
\end{equation}
 where $G_t$ is given by \eqref{eq:Gs} and  $\psi$ is the solution of the following Riccati Volterra equation 
\begin{align}
	\psi(t)&=\int_0^t K(t-s)  F(\psi(s))ds, \quad t \geq 0, \label{eq:Ric1} \\
	F(u)&= w + \frac 12 (v^2 -v) +  (\rho c v  +b)u + \frac {c^2} 2 u^2,  \label{eq:Ric2}
	\end{align}
    and satisfies $      \Re \psi \leq 0. $  
    
We refer to \cite[Theorem 7.3]{abi2021weak} for the case of $L^1_{\rm loc}$ kernels, and \cite[Theorem 7.1]{abijaber2019affine}  for the equivalent formulations in the specific (and more standard) case of $L^2_{\rm loc}$ kernels.

\bibliographystyle{plainnat}
\bibliography{bibl}

\end{document}